\documentclass[a4paper,UKenglish,cleveref,hyperref,autoref]{lipics-v2021}

\title{Approximating Highly Inapproximable Problems on Graphs of Bounded Twin-Width}

\titlerunning{Approximating Highly Inapproximable Problems on Graphs of Bounded Twin-Width}

\author{Pierre Bergé}{Univ Lyon, CNRS, ENS de Lyon, Université Claude Bernard Lyon 1, LIP UMR5668, France}{pierre.berge@ens-lyon.fr}{}{}
\author{\'{E}douard Bonnet}{Univ Lyon, CNRS, ENS de Lyon, Université Claude Bernard Lyon 1, LIP UMR5668, France \and \url{http://perso.ens-lyon.fr/edouard.bonnet/}}{edouard.bonnet@ens-lyon.fr}{https://orcid.org/0000-0002-1653-5822}{}
\author{Hugues Déprés}{Univ Lyon, CNRS, ENS de Lyon, Université Claude Bernard Lyon 1, LIP UMR5668, France \and \url{http://perso.ens-lyon.fr/hugues.depres/}}{hugues.depres@ens-lyon.fr}{}{}
\author{Rémi Watrigant}{Univ Lyon, CNRS, ENS de Lyon, Université Claude Bernard Lyon 1, LIP UMR5668, France}{remi.watrigant@ens-lyon.fr}{}{}

\authorrunning{P. Bergé, \'E. Bonnet, H. Déprés, R. Watrigant}

\Copyright{Pierre Bergé, Édouard Bonnet, Hugues Déprés, Rémi Watrigant}



\category{}

\relatedversion{}

\supplement{}

\acknowledgements{We thank Colin Geniet, Eunjung Kim, and Stéphan Thomassé for useful discussions.}

\nolinenumbers 

\hideLIPIcs  

\EventEditors{John Q. Open and Joan R. Access}
\EventNoEds{2}
\EventLongTitle{42nd Conference on Very Important Topics (CVIT 2016)}
\EventShortTitle{CVIT 2016}
\EventAcronym{CVIT}
\EventYear{2016}
\EventDate{December 24--27, 2016}
\EventLocation{Little Whinging, United Kingdom}
\EventLogo{}
\SeriesVolume{42}
\ArticleNo{23}

\usepackage[utf8]{inputenc}  

\usepackage[T1]{fontenc}
\usepackage{lmodern}

\usepackage[colorinlistoftodos,bordercolor=orange,backgroundcolor=orange!20,linecolor=orange,textsize=normalsize]{todonotes}

\usepackage{amsmath}  
\usepackage{amssymb}     
\usepackage{bbm}
\usepackage{accents}
\usepackage{complexity}
\usepackage{booktabs}
\usepackage{fixmath}
\usepackage{paralist}

\makeatletter
\newtheorem*{rep@theorem}{\rep@title}
\newcommand{\newreptheorem}[2]{%
\newenvironment{rep#1}[1]{%
 \def\rep@title{#2 \ref{##1}}%
 \begin{rep@theorem}}%
 {\end{rep@theorem}}}
\makeatother

\newreptheorem{theorem}{Theorem}
\newreptheorem{lemma}{Lemma}

\crefname{observation}{Observation}{Observations}

\usepackage{xspace}
\usepackage{tikz}

\usepackage[ruled,vlined,linesnumbered]{algorithm2e}

\usetikzlibrary{fit}
\usetikzlibrary{arrows}
\usetikzlibrary{patterns}
\usetikzlibrary{calc}
\usetikzlibrary{shapes}
\usetikzlibrary{positioning}
\usetikzlibrary{math}
\usetikzlibrary{shapes.symbols}
\usetikzlibrary{decorations.pathreplacing}
\tikzset{draw half paths/.style 2 args={%
  decoration={show path construction,
    lineto code={
      \draw [#1] (\tikzinputsegmentfirst) -- 
         ($(\tikzinputsegmentfirst)!0.5!(\tikzinputsegmentlast)$);
      \draw [#2] ($(\tikzinputsegmentfirst)!0.5!(\tikzinputsegmentlast)$)
        -- (\tikzinputsegmentlast);
    }
  }, decorate
}}

\usepackage[scr=boondox,scrscaled=1.05]{mathalfa}

\renewcommand{\leq}{\leqslant}

\renewcommand{\le}{\leq}

\newcommand{\mds}{\textsc{Min Dominating Set}\xspace}

\newcommand{\mis}{\textsc{Max Independent Set}\xspace}
\newcommand{\wis}{\textsc{Weighted Max Independent Set}\xspace}
\newcommand{\smis}{\textsc{MIS}\xspace}
\newcommand{\swis}{\textsc{WMIS}\xspace}

\newcommand{\msim}{\textsc{Max Subset Induced Matching}\xspace}
\newcommand{\mim}{\textsc{Max Induced Matching}\xspace}

\newcommand{\mids}{\textsc{Min Independent Dominating Set}\xspace}

\newcommand{\lpath}{\textsc{Longest Path}\xspace}
\newcommand{\ipath}{\textsc{Longest Induced Path}\xspace}

\newcommand{\mif}{\textsc{Max Edge Induced Forest}\xspace}
\newcommand{\misf}{\textsc{Max Edge Induced Star Forest}\xspace}
\newcommand{\mlisf}{\textsc{Max Leaves Induced Star Forest}\xspace}

\newcommand{\mihp}[1]{\textsc{Mutually Induced $#1$-packing}\xspace}
\newcommand{\aihp}[1]{\textsc{Annotated Mutually Induced $#1$-packing}\xspace}

\newcommand{\mihph}{\textsc{Mutually Induced $H$-packing}\xspace}

\newcommand{\coloring}{\textsc{Coloring}\xspace}
\newcommand{\setcoloring}{\textsc{Set Coloring}\xspace}

\newcommand{\regu}{cleanup\xspace}

\newcommand{\fullregu}{full cleanup\xspace}
\newcommand{\fullregus}{full cleanups\xspace}

\newcommand{\defoptproblem}[3]{
 \vspace{1mm}
\noindent\fbox{
 \begin{minipage}{0.96\textwidth}
 #1 \\
 {\bf{Input:}} #2 \\
 {\bf{Output:}} #3
 \end{minipage}
 }
 \vspace{1mm}
}

\theoremstyle{definition}

\newcommand{\tww}{\text{tww}}
\renewcommand{\P}{{\mathcal P}}

\renewcommand{\R}{\mathbb R}
\newcommand{\Q}{\mathbb Q}

\begin{document}

\maketitle

\begin{abstract}
 For any $\varepsilon > 0$, we give a polynomial-time $n^\varepsilon$-approximation algorithm for \textsc{Max Independent Set} in graphs of bounded twin-width given with an $O(1)$-sequence.
  This result is derived from the following time-approximation trade-off: We establish an $O(1)^{2^q-1}$-approximation algorithm running in time $\exp(O_q(n^{2^{-q}}))$, for every integer $q \geqslant 0$.
  Guided by the same framework, we obtain similar approximation algorithms for \textsc{Min Coloring} and \textsc{Max Induced Matching}.
  In general graphs, all these problems are known to be highly inapproximable: for any $\varepsilon > 0$, a~polynomial-time $n^{1-\varepsilon}$-approximation for any of them would imply that P$=$NP [H{\aa}stad, FOCS '96; Zuckerman, ToC '07; Chalermsook et al., SODA '13].
  We generalize the algorithms for \textsc{Max Independent Set} and \textsc{Max Induced Matching} to the independent (induced) packing of any fixed connected graph $H$.
  
  In contrast, we show that such approximation guarantees on graphs of bounded twin-width given with an $O(1)$-sequence are very unlikely for \textsc{Min Independent Dominating Set}, and somewhat unlikely for \textsc{Longest Path} and \textsc{Longest Induced Path}.
  Regarding the existence of better approximation algorithms, there is a (very) light evidence that the obtained approximation factor of $n^\varepsilon$ for \textsc{Max Independent Set} may be best possible.
  This is the first in-depth study of the approximability of problems in graphs of bounded twin-width.
  Prior to this paper, essentially the only such result was a~polynomial-time $O(1)$-approximation algorithm for \textsc{Min Dominating Set} [Bonnet et al., ICALP '21].
\end{abstract}

\section{Introduction}

Twin-width is a graph parameter introduced by Bonnet, Kim, Thomassé, and Watrigant~\cite{twin-width1}.
Its definition involves the notions of \emph{trigraphs} and of \emph{contraction sequences}.
A~\emph{trigraph} is a graph with two types of edges: black (regular) edges and red (error) edges.
A~(vertex) \emph{contraction} consists of merging two (non-necessarily adjacent) vertices, say, $u, v$ into a~vertex~$w$, and keeping every edge $wz$ black if and only if $uz$ and $vz$ were previously black edges.
The other edges incident to $w$ become red (if not already), and the rest of the trigraph remains the same.
A~\emph{contraction sequence} of an $n$-vertex\footnote{In this introduction, we might implicitly use $n$ to denote the number of vertices, and $m$, the number of edges of the graph at hand.} graph $G$ is a sequence of trigraphs $G=G_n,$ $\ldots, G_1=K_1$ such that $G_i$ is obtained from $G_{i+1}$ by performing one contraction.
A~\mbox{\emph{$d$-sequence}} is a contraction sequence in which every vertex of every trigraph has at most $d$ red edges incident to it.
The~\emph{twin-width} of $G$, denoted by $\tww(G)$, is then the minimum integer~$d$ such that $G$ admits a $d$-sequence.
\Cref{fig:contraction-sequence} gives an example of a graph with a 2-sequence, i.e., of twin-width at most~2.
Twin-width can be naturally extended to matrices (with unordered~\cite{twin-width1} or ordered~\cite{twin-width4} row and column sets) over a finite alphabet, and thus to binary structures.

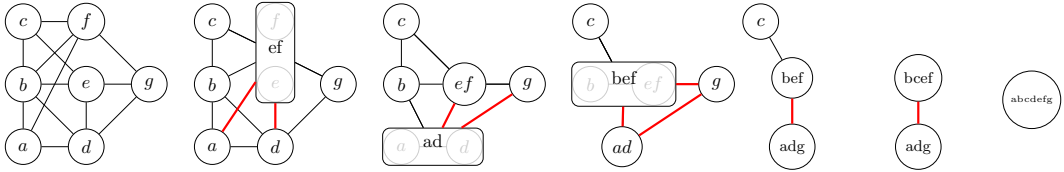
\begin{figure}[h!]
  \centering
  \resizebox{400pt}{!}{
  \begin{tikzpicture}[
      vertex/.style={circle, draw, minimum size=0.68cm}
    ]
    \def\s{1.2}
    \foreach \i/\j/\l in {0/0/a,0/1/b,0/2/c,1/0/d,1/1/e,1/2/f,2/1/g}{
      \node[vertex] (\l) at (\i * \s,\j * \s) {$\l$} ;
    }
    \foreach \i/\j in {a/b,a/d,a/f,b/c,b/d,b/e,b/f,c/e,c/f,d/e,d/g,e/g,f/g}{
      \draw (\i) -- (\j) ;
    }

    \begin{scope}[xshift=3 * \s cm]
    \foreach \i/\j/\l in {0/0/a,0/1/b,0/2/c,1/0/d,2/1/g}{
      \node[vertex] (\l) at (\i * \s,\j * \s) {$\l$} ;
    }
    \foreach \i/\j/\l in {1/1/e,1/2/f}{
      \node[vertex,opacity=0.2] (\l) at (\i * \s,\j * \s) {$\l$} ;
    }
    \node[draw,rounded corners,inner sep=0.01cm,fit=(e) (f)] (ef) {ef} ;
    \foreach \i/\j in {a/b,a/d,b/c,b/d,b/ef,c/ef,c/ef,d/g,ef/g,ef/g}{
      \draw (\i) -- (\j) ;
    }
    \foreach \i/\j in {a/ef,d/ef}{
      \draw[red, very thick] (\i) -- (\j) ;
    }
    \end{scope}

    \begin{scope}[xshift=6 * \s cm]
    \foreach \i/\j/\l in {0/1/b,0/2/c,2/1/g,1/1/ef}{
      \node[vertex] (\l) at (\i * \s,\j * \s) {$\l$} ;
    }
    \foreach \i/\j/\l in {0/0/a,1/0/d}{
      \node[vertex,opacity=0.2] (\l) at (\i * \s,\j * \s) {$\l$} ;
    }
    \draw[opacity=0.2] (a) -- (d) ;
    \node[draw,rounded corners,inner sep=0.01cm,fit=(a) (d)] (ad) {ad} ;
    \foreach \i/\j in {ad/b,b/c,b/ad,b/ef,c/ef,c/ef,ef/g,ef/g}{
      \draw (\i) -- (\j) ;
    }
    \foreach \i/\j in {ad/ef,ad/g}{
      \draw[red, very thick] (\i) -- (\j) ;
    }
    \end{scope}

    \begin{scope}[xshift=9 * \s cm]
    \foreach \i/\j/\l in {0/2/c,2/1/g,0.5/0/ad}{
      \node[vertex] (\l) at (\i * \s,\j * \s) {$\l$} ;
    }
    \foreach \i/\j/\l in {0/1/b,1/1/ef}{
      \node[vertex,opacity=0.2] (\l) at (\i * \s,\j * \s) {$\l$} ;
    }
    \draw[opacity=0.2] (b) -- (ef) ;
    \node[draw,rounded corners,inner sep=0.01cm,fit=(b) (ef)] (bef) {bef} ;
    \foreach \i/\j in {ad/bef,bef/c,bef/ad,c/bef,c/bef,bef/g}{
      \draw (\i) -- (\j) ;
    }
    \foreach \i/\j in {ad/bef,ad/g,bef/g}{
      \draw[red, very thick] (\i) -- (\j) ;
    }
    \end{scope}

    \begin{scope}[xshift=11.7 * \s cm]
    \foreach \i/\j/\l in {0/2/c}{
      \node[vertex] (\l) at (\i * \s,\j * \s) {$\l$} ;
    }
     \foreach \i/\j/\l in {0.5/0/adg,0.5/1.1/bef}{
      \node[vertex] (\l) at (\i * \s,\j * \s) {\footnotesize{\l}} ;
    }
    \foreach \i/\j in {c/bef}{
      \draw (\i) -- (\j) ;
    }
    \foreach \i/\j in {adg/bef}{
      \draw[red, very thick] (\i) -- (\j) ;
    }
    \end{scope}

    \begin{scope}[xshift=13.7 * \s cm]
    \foreach \i/\j/\l in {0.5/0/adg,0.5/1.1/bcef}{
      \node[vertex] (\l) at (\i * \s,\j * \s) {\footnotesize{\l}} ;
    }
    \foreach \i/\j in {adg/bcef}{
      \draw[red, very thick] (\i) -- (\j) ;
    }
    \end{scope}

    \begin{scope}[xshift=15 * \s cm]
    \foreach \i/\j/\l in {1/0.75/abcdefg}{
      \node[vertex] (\l) at (\i * \s,\j * \s) {\tiny{\l}} ;
    }
    \end{scope}
    
  \end{tikzpicture}
  }
  \caption{A 2-sequence witnessing that the initial graph has twin-width at most~2.}
  \label{fig:contraction-sequence}
\end{figure}

An equivalent viewpoint that will be somewhat more convenient is to consider a $d$-sequence as a sequence of partitions $\P_n := \{\{v\}~:~v \in V(G)\}, \P_{n-1}, \ldots, \P_1 := \{V(G)\}$ of $V(G)$, such that for every integer $1 \leqslant i \leqslant n-1$, $\P_i$ has $i$ parts and is obtained by merging two parts of $\P_{i+1}$ into one.
Now the \emph{red degree} of a part $P \in \P_i$ is the number of other parts $Q \in \P_i$ such that there is in $G$ at least one edge and at least one non-edge between $P$ and~$Q$.
A~$d$-sequence is such that no part of no partition of the sequence has red degree more than~$d$.
In that case the \emph{maximum red degree} of each partition is at most~$d$.  
And we similarly get the twin-width of $G$ as the minimum integer $d$ such that $G$ admits a~(partition) $d$-sequence.
The~\emph{quotient trigraph $G/\P_i$} is the trigraph $G_i$, if the (contraction) $d$-sequence $G_n, \ldots, G_1$ and the (partition) $d$-sequence $\P_n, \ldots, \P_1$ correspond.

Classes of binary structures with bounded twin-width include graph classes with bounded treewidth, and more generally bounded clique-width, proper minor-closed classes, posets with antichains of bounded size, strict subclasses of permutation graphs, as well as $\Omega(\log n)$-subdivisions of $n$-vertex graphs~\cite{twin-width1}, and some classes of (bounded-degree) expanders~\cite{twin-width2}.
A~notable variety of geometrically defined graph classes have bounded twin-width such as map graphs, bounded-degree string graphs~\cite{twin-width1}, classes with bounded queue number or bounded stack number~\cite{twin-width2}, segment graphs with no $K_{t,t}$ subgraph, visibility graphs of 1.5D terrains without large half-graphs, visibility graphs of simple polygons without large independent sets~\cite{twin-width8}.

For every class $\mathcal C$ mentioned so far, $O(1)$-sequences can be computed in polynomial time\footnote{Admittedly, for the geometric classes, a representation is (at least partially) needed.} on members of~$\mathcal C$.
For classes of binary structures including a binary relation interpreted as a linear order on the domain (called \emph{ordered binary structures}), there is a fixed-parameter approximation algorithm for twin-width~\cite{twin-width4}.
More precisely, given a graph~$G$ and an integer~$k$, there are computable functions $f$ and $g$ such that one can output an~$f(k)$-sequence of $G$ or correctly report that $\tww(G)>k$ in time $g(k)n^{O(1)}$.   
Such an approximation algorithm is currently missing for classes of general (not necessarily ordered) binary structures, and in particular for the class of all graphs.
We also observe that deciding if the twin-width of a~graph is at most~4 is an NP-complete task~\cite{Berge21}.

We will therefore assume that the input graph is given with a~$d$-sequence, and treat~$d$ as a constant (or that the input comes from any of the above-mentioned classes).
Thus far, this is the adopted setting when designing faster algorithms on bounded twin-width graphs~\cite{twin-width1,twin-width3,PilipczukSZ22,Kratsch22,Gajarsky22}.
From the inception of twin-width~\cite{twin-width1} --actually already from the seminal work of Guillemot and Marx~\cite{Guillemot14}-- it was clear that structures wherein this invariant is bounded may \emph{often} allow the design of parameterized algorithms.
More concretely, it was shown~\cite{twin-width1} that, on graphs $G$ given with a~$d$-sequence, model checking a~first-order sentence $\varphi$ is fixed-parameter tractable --it can be solved in time $f(d,\varphi) \cdot n$--, the special cases of, say, \textsc{$k$-Independent Set} or \textsc{$k$-Dominating Set} admit single-exponential parameterized algorithms~\cite{twin-width3}, an effective data structure almost linear in $n$ can support constant-time edge queries~\cite{PilipczukSZ22}, the triangles of $G$ can be counted in time $O(d^2n+m)$~\cite{Kratsch22}.

So far, however, the connection between \emph{having bounded twin-width} and \emph{enjoying enhanced approximation factors} was tenuous.
The only such result concerned \textsc{Min Dominating Set}, known to be inapproximable in polynomial-time within factor $(1-o(1))\ln n$ unless P$=$NP~\cite{Dinur14}, but yet admits a constant-approximation on graphs of bounded twin-width given with an $O(1)$-sequence~\cite{twin-width3}.
We start filling this gap by designing approximation algorithms on graphs of bounded twin-width given with an $O(1)$-sequence for notably \mis (\smis, for short), \mim, and \coloring.
Getting better approximation algorithms for \smis and \coloring in that particular scenario was raised as an open problem~\cite{twin-width3}.
Before we describe our results and elaborate on the developed techniques, let us briefly present the notorious inapproximability of these problems in general graphs.

\smis and \coloring are NP-hard~\cite{GJ79}, and very inapproximable: for every $\varepsilon > 0$, it is NP-hard to approximate these problems within ratio $n^{1-\varepsilon}$ \cite{Hastad96,Zuckerman07}.
The same was shown to hold for \mim~\cite{Chalermsook13b}.
Besides, there is only little room to improve over the brute-force algorithm in $2^{O(n)}$: Unless the Exponential Time Hypothesis\footnote{That is, the assumption that there is a $\delta > 0$ such that $n$-variable \textsc{3-SAT} cannot be solved in time $\delta^n$.}~\cite{Impagliazzo01} (ETH) fails, no algorithm can solve \smis in time $2^{o(n)}$~\cite{sparsification} (nor the other two problems).
For any~$r$ (possibly a function of $n$) \swis can be $r$-approximated in time $2^{O(n/r)}$~\cite{Cygan08,Bourgeois11}.
Bansal et al.~\cite{Bansal19} essentially shaved a $\log^2 r$ factor to the latter exponent.
It is known though that polynomial shavings are unlikely.
Chalermsook et al.~\cite{Chalermsook13} showed that, for any $\varepsilon > 0$ and sufficiently large $r$ (again $r$ can be function of $n$), an $r$-approximation for \smis and \mim cannot take time $2^{O(n^{1-\varepsilon}/r^{1+\varepsilon})}$, unless the ETH fails.
For instance, investing time $2^{O(\sqrt n)}$, one cannot hope for significantly better than a $\sqrt n$-approximation.

\paragraph*{Contributions and techniques}

Our starting point is a constant-approximation algorithm for \smis running in time $2^{O(\sqrt n)}$ when presented with an $O(1)$-sequence, which is very unlikely to hold in general graphs by the result of~Chalermsook et al.~\cite{Chalermsook13}.

\begin{theorem}\label{thm:intro-mis}
  On $n$-vertex graphs given with a $d$-sequence \mis can be $O_d(1)$-approximated in time $2^{O_d(\sqrt n)}$.
\end{theorem}

Our algorithm builds upon the functional equivalence between twin-width and the so-called \emph{versatile twin-width}~\cite{twin-width2}.
We defer the reader to~\cref{sec:prelim} for a formal definition of versatile twin-width.
For our purpose, one only needs to know the following useful consequence of that equivalence.
From a $d'$-sequence of $G$, we can compute in polynomial time another partition sequence $\P_n, \ldots, \P_1$ of $G$ of width $d := f(d')$, for some computable function $f$, such that for every integer $1 \leqslant i \leqslant n$, all the $i$ parts of $\P_i$ have size at most $d \cdot \frac{n}{i}$.
Even if some parts of $P_i$ can be very small, this partition is balanced in the sense that no part can be larger than $d$ times the part size in a perfectly balanced partition.
Of importance to us is $\P_{\lfloor \sqrt n \rfloor}$ when the number of parts ($\lfloor \sqrt n \rfloor$) and the size of a larger part in the partition (at most $d \frac{n}{\lfloor \sqrt n \rfloor} \approx d \sqrt n$) are somewhat level.

We can then properly color the red graph (made by the red edges on the vertex set $\P_{\lfloor \sqrt n \rfloor}$) with $d+1$ colors.
Any color class $X$ is a subset of parts of $\P_{\lfloor \sqrt n \rfloor}$ such that between two parts there are either all edges (black edge) or no edge at all (non-edge).
In graph-theoretic terms, the subgraph $G_X$ of $G$ induced by all the vertices of all the parts of $X$ have a simple modular decomposition: a partition of at most $\sqrt n$ modules each of size at most $d \sqrt n$.
It is thus routine to compute a largest independent set of $G_X$ essentially in time exponential in the maximum between the number of modules and the maximum size of a module, that is, in at most $d \sqrt n$.
As one color class $X^*$ contains more than a $\frac{1}{d+1}$ fraction of the optimum, we get our $d+1$-approximation when computing a largest independent set of $G_{X^*}$.
\Cref{fig:alg-mis} serves as a visual summary of what we described so far. 

The next step is to substitute recursive calls of our approximation algorithm to exact exponential algorithms on induced subgraphs of size $O_d(\sqrt n)$.
Following this inductive process at depth $q=2, 3, 4, \ldots$, we degrade the approximation ratio to $(d+1)^3, (d+1)^7, (d+1)^{15}$, etc. but meanwhile we boost the running time to $2^{O_d(n^{1/4})}, 2^{O_d(n^{1/8})}, 2^{O_d(n^{1/16})}$, etc. 
In effect we show by induction that:
\begin{theorem}\label{thm:intro-mis2}
  On $n$-vertex graphs given with a $d$-sequence \mis has an $O_d(1)^{2^q-1}$-approximation algorithm running in time $2^{O_{d,q}(n^{2^{-q}})}$, for every integer $q \geqslant 0$.
\end{theorem}

The following polynomial-time algorithm is a corollary of~\cref{thm:intro-mis2} choosing $q = O_{d,\varepsilon}(\log \log n)$.

\begin{theorem}\label{thm:intro-mis3}
  For every $\varepsilon > 0$, \mis can be $n^\varepsilon$-approximated in polynomial-time $O_{d,\varepsilon}(1) \cdot \log^{O_d(1)}n \cdot n^{O(1)}$ on $n$-vertex graphs given with a $d$-sequence.
\end{theorem}
Note that the exponent of the polynomial factor is an absolute constant (not depending on $d$ nor on $\varepsilon$).

\medskip

We then apply our framework to \coloring and \mim.
\begin{theorem}\label{thm:intro-col-mim}
   For every $\varepsilon > 0$, \coloring and \mim admit polynomial-time $n^\varepsilon$-approximation algorithms on $n$-vertex graphs of bounded twin-width given with an $O(1)$-sequence.
\end{theorem}

The main additional difficulty for \coloring is that one cannot satisfactorily solve/approximate that problem on a modular decomposition by simply coloring its modules and its quotient graph.
One needs to tackle a more general problem called \setcoloring.
Fortunately this generalization is the fixed point we are looking for: approximating \setcoloring can be done in our framework by mere recursive calls (to itself).

For \mim, we face a new kind of obstacle.
It can be the case that no decent solution is contained in any color class $X$ --in the chosen $d+1$-coloring of the red graph $G/\P_{\lfloor \sqrt n \rfloor}$.
For instance, it is possible that any such color class $X$ induces in $G$ an edgeless graph, while very large induced matchings exist with endpoints in two distinct color classes.
We thus need to also find large induced matchings within the black edges and within the red edges of~$G/\P_{\lfloor \sqrt n \rfloor}$.
This leads to a more intricate strategy intertwining the coloring of bounded-degree graphs (specifically the red graph and the square of its line graph) and recursive calls to induced subgraphs of $G$, and to special induced subgraphs of the total graph (i.e., made by both the red and black edges) of~$G/\P_{\lfloor \sqrt n \rfloor}$.
Although this is not necessary, one can observe that the latter graphs are also induced subgraphs of $G$ itself.

\medskip

We then explore the limits of our results and framework in terms of amenable problems. 
We give the following technical generalization to the approximation algorithms for \smis and \mim.
\begin{theorem}\label{thm:intro-mihp}
   For every connected graph $H$ and $\varepsilon > 0$, \mihp{H} admits a~polynomial-time $n^\varepsilon$-approximation algorithms on $n$-vertex graphs of bounded twin-width given with an $O(1)$-sequence.
\end{theorem}

In this problem, one seeks for a largest induced subgraph that consists of a disjoint union of copies of $H$.
All the previous technical issues are here combined.
We try all the possibilities of batching the vertices of $H$ into at most $|V(H)|$ parts of~$G/\P_{\lfloor \sqrt n \rfloor}$, based on the trigraph that these parts define.
For instance with $H=K_2$ (an edge), i.e., the case of \mim, the three possible trigraphs are the 1-vertex trigraph, two vertices linked by a red edge, and two vertices linked by a black edge.
In the general case, the problem generalization is quite delicate to find.
We have to keep some partitions of $V(G)$ and $V(H)$ to enforce that the copies of $H$ in $G$ follow a pattern that the algorithm committed to higher up in the recursion tree, and a weight function on $|V(H)|$-tuples of vertices of $G$, not to forget how many mutually induced copies of $H$ can be packed \emph{within} these vertices.
The other novelty is that some recursive calls are on induced subgraphs of the total graph of~$G/\P_{\lfloor \sqrt n \rfloor}$ that are \emph{not} induced subgraphs of $G$.
Fortunately, these graphs keep the same bound of versatile twin-width, and thus our framework allows it. 

Defining, for a family of graphs $\mathcal H$, \mihp{\mathcal H} as the same problem where the connected components of the induced subgraph should all be in $\mathcal H$, we get a similar approximation factor when $\mathcal H$ is a finite set of connected graphs.
(Note that \mihp{H} is sometimes called \textsc{Independent Induced $H$-Packing}.)
In particular, we can similarly approximate \textsc{Independent $H$-Packing}, which is the same problem but the copies of $H$ need not be induced.
(Our approximation algorithms could extend to other $H$-packing variants without the independence requirement, but these problems can straightforwardly be $O(1)$-approximated in general graphs.)

We can handle some cases when $\mathcal H$ is infinite, too.
For instance, by slightly adapting the case of \smis, we can get an $n^\varepsilon$-approximation when $\mathcal H$ is the set of all cliques.
We show this more involved example, also expressible as $\mihp{\mathcal H}$ for $\mathcal H$ the set of all trees or the set all stars.
\begin{theorem}\label{thm:intro-star-forest}
   For every $\varepsilon > 0$, finding the induced (star) forest with the most edges admits a~polynomial-time $n^\varepsilon$-approximation algorithms on $n$-vertex graphs of bounded twin-width given with an $O(1)$-sequence. 
\end{theorem}

As we already mentioned, our framework is exclusively useful for problems that are very inapproximable in general graphs; at least for which an $n^{\varepsilon}$-approximation algorithm is not known for every $\varepsilon > 0$.
Are there natural such problems that cannot be approximated better in graphs of bounded twin-width?
We answer this question positively with the example of \mids.

\begin{theorem}\label{thm:intro-inapprox-mids}
For every $\varepsilon > 0$, \mids does not admit an $n^{1-\varepsilon}$-approximation algorithm in $n$-vertex graphs given with an $O(1)$-sequence, unless P$=$NP.   
\end{theorem}

The reduction is the same as the one for general graphs~\cite{Halldorsson93}, but performed from a planar variant of \textsc{3-SAT}.
The obtained instances are not planar but can be contracted to planar trigraphs, hence overall have bounded twin-width.

Finally the case of \lpath and \ipath is interesting.
The best approximation factor for the former~\cite{Gabow08} is worse than $n^{0.99}$, while the latter is known to have the same inapproximability as \smis~\cite{Lund93}.
However an $n^\varepsilon$-approximation algorithm (for every $\varepsilon>0$) is not excluded for \lpath.
We show that the property of bounded twin-width is unlikely to help for these two problems, as it would lead to better approximation algorithms for \lpath in general graphs.
This is mainly because subdividing at least $2 \log n$ times every edge of any $n$-vertex graph gives a graph with twin-width at most~4~\cite{Berge21}. 

\begin{theorem}\label{thm:intro-inapprox-path}
   For any $r=\omega(1)$, an $r$-approximation for \ipath or \lpath on graphs given with an $O(1)$-sequence would imply a~$(1+o(1))r$-approximation for \lpath in general graphs.
\end{theorem}

In turn, this can be used to exhibit a family $\mathcal H$ with an infinite antichain for the \emph{induced subgraph} relation such that \mihp{\mathcal H} is \emph{hard} to $n^\varepsilon$-approximate on graphs of bounded twin-width.
The family $\mathcal H$ is simply the set of all paths terminated by triangles at both ends.

\begin{theorem}\label{thm:intro-inapprox-mihp-dec-paths}
 There is an infinite family $\mathcal H$ of connected graphs such that if for every $\varepsilon >0$, \mihp{\mathcal H} admits an $n^\varepsilon$-approximation algorithm on $n$-vertex graphs given with an $O(1)$-sequence, then so does \lpath on general graphs. 
\end{theorem}

\cref{tbl:app-problems} summarizes our results and hints at future work.

\begin{table}[h!]
  \centering
\begin{tabular}{lcccc}
  \toprule
  ~~~~Problem name &  lower bound       &  upper bound                                                      & lower bound      \\
                   &  general graphs    &   bounded $\tww$                                                & bounded $\tww$ \\
  \midrule
  \mis    &  $n^{1-\varepsilon}$  & \framebox{$n^\varepsilon$}                                           & ?, self-improvement \\
  \coloring &  $n^{1-\varepsilon}$  & \framebox{$n^\varepsilon$}                                           & $4/3-\varepsilon$ \\
  \mim      & $n^{1-\varepsilon}$  & \framebox{$n^\varepsilon$}                                           & ? \\
  \textsc{Mut. Ind. $H$-Packing} & $n^{1-\varepsilon}$  & \framebox{$n^\varepsilon$ ($H$ connected)}      & ? \\
  \textsc{Mut. Ind. $\mathcal H$-Packing}      & $n^{1-\varepsilon}$  & \framebox{$n^\varepsilon$ for some $\mathcal H$}& \framebox{\lpath-hard} \\
  \textsc{Min Ind. Dom. Set}     & $n^{1-\varepsilon}$   & $n/\polylog(n)$                              & \framebox{$n^{1-\varepsilon}$} \\
  \lpath &  $2^{\log^{1-\varepsilon}n}$     & $n/\exp(\Omega(\sqrt{\log n}))$                           & \framebox{\lpath-hard}   \\
  \ipath &  $n^{1-\varepsilon}$                      & $n/\polylog(n)$                              & \framebox{\lpath-hard} \\
  \mds   &  $(1-\varepsilon)\ln n$ & $O(1)$ & ? \\
  \bottomrule
\end{tabular}
\caption{Approximability status of graph problems in general graphs and in graphs of bounded twin-width given with an $O(1)$-sequence.
  Everywhere ``$\varepsilon$'' should be read as ``$\forall \varepsilon > 0$''.
  Our results are enclosed by boxes.
  ``\lpath-hard'' means that getting an $r$-approximation would yield essentially the same ratio for \lpath in general graphs.
  The other lower bounds are under standard complexity-theoretic assumptions, mostly P$\neq$NP.
  Not to clutter the table, we do not put the references, which can all be found in the paper.}
\label{tbl:app-problems}
\end{table}

For the main highly inapproximable graph problems, we either obtain an $n^\varepsilon$-approximation algorithm on graphs of bounded twin-width given with an $O(1)$-sequence, or a conditional obstruction to such an algorithm.
In the former case, can we improve further the approximation factor?
The next theorem was observed using the self-improvement reduction of~Feige et al.~\cite{Feige91}, which preserves the twin-width bound.
This reduction consists of going from a graph $G$ to the lexicographic product $G[G]$, where every vertex of $G$ is replaced by a module inducing a copy of $G$ (and iterating this trick).

\begin{theorem}[\cite{twin-width3}]\label{thm:intro-self-improvement}
  Let $r: \mathbb N \to \mathbb R$ be any non-decreasing function such that for every $\varepsilon>0$, $r(n)=o(n^\varepsilon)$.
  If \mis admits an $r(n)$-approximation algorithm on $n$-vertex graphs of bounded twin-width given with an $O(1)$-sequence, then it further admits an $r(n)^\varepsilon$-approximation.
\end{theorem}

To our knowledge, the application of the self-improvement trick is always to strengthen a~lower bound, and never to effortlessly obtain a~better approximation factor.
Therefore, we may take~\cref{thm:intro-self-improvement} as a weak indication that our approximation ratio is best possible.
Still, not even a~polynomial-time approximation scheme (PTAS) is ruled out for \smis (nor for \mim, \mds, etc.) and we would like to see better approximation algorithms. 
For \coloring, as was previously observed~\cite{twin-width3}, a~PTAS is ruled out by the NP-hardness of deciding if a planar graph is 3-colorable or 4-chromatic, since planar graphs have twin-width at most~9 and a 9-sequence can be found in linear time~\cite{planar-tww}.

\section{Preliminaries}\label{sec:prelim}

For $i$ and $j$ two integers, we denote by $[i,j]$ the set of integers that are at least $i$ and at most~$j$.
For every integer $i$, $[i]$ is a shorthand for $[1,i]$.

\subsection{Handled graph problems}

We will consider several problems throughout the paper.
We recall here the definition of the most central ones.
Some technical problem generalizations will be defined along the way.

\defoptproblem{\wis (\swis, for short)}{A graph $G$ and a weight function $V(G) \to \Q$.}{A set $S \subseteq V(G)$ such that $\forall u,v \in S$, $uv \notin E(G)$ maximizing $w(S) := \sum\limits_{v \in S}w(v)$.}

A feasible solution to \swis is called an \emph{independent set}.
The \mis (\smis, for short) problem is the particular case with $w(v)=1$, $\forall v \in V(G)$.
We may denote by $\alpha(G)$, the \emph{independence number}, that is the optimum value of \swis on graph $G$.

\defoptproblem{\coloring}{A graph $G$.}{A partition $\P$ of $V(G)$ into independent sets minimizing the cardinality of~$\P$.}

Equivalently, \coloring can be expressed as finding an integer $k$ and a~map $c: V(G) \to [k]$ such that for every $uv \in E(G)$, $c(u) \neq c(v)$, while minimizing $k$.

\defoptproblem{\mim}{A graph $G$, possibly together with a weight function $w: E(G) \to \Q$.}{A set $S \subseteq E(G)$ such that $\forall uv \neq u'v' \in S$, $\{u,v\} \cap \{u',v'\} = \emptyset$ and $G[\{u,v,u',v'\}]$ has exactly two edges, maximizing $w(S) := \sum\limits_{e \in S}w(e)$.}

An~\emph{induced matching} is a~pairwise disjoint set of edges (i.e., a~matching) with no edge bridging them.
We now give a common generalization of \swis and \mim.

\defoptproblem{\mihp{\mathcal H}}{A graph $G$, possibly together with a weight function $w: V(G) \to \Q$.}{A set $S \subseteq V(G)$ such that $G[S]$ is a disjoint union of graphs each isomorphic to a graph in $\mathcal H$, maximizing $w(S) := \sum\limits_{v \in S}w(v)$.}

When $\mathcal H$ consists of a single graph, say $H$, we simply denote the former problem \mihp{H}.
\swis and \mim are the special cases when $H$ is a vertex and an edge, respectively.

\subsection{The contraction and partition viewpoints of twin-width}

A~\emph{trigraph} $G$ has vertex set $V(G)$, black edge set $E(G)$, red edge set $R(G)$ such that $E(G) \cap R(G) = \emptyset$ (and $E(G),R(G) \subseteq {V(G) \choose 2}$).
A~\emph{contraction} in a trigraph $G$ replaces a pair of (non-necessarily adjacent) vertices $u, v \in V(G)$ by one vertex $w$ that is linked to $G-\{u,v\}$ in the following way to form a new trigraph $G'$.
For every $z \in V(G) \setminus \{u,v\}$, $wz \in E(G')$ whenever $uz, vz \in E(G)$, $wz \notin E(G') \cup R(G')$ whenever $uz, vz \notin E(G) \cup R(G)$, and $wz \in R(G')$, otherwise.
The \emph{red graph} $(V(G),R(G))$ will be denoted by $\mathcal R(G)$.
We denote by $\mathcal T(G)$ the \emph{total graph} of $G$ defined as $(V(G),E(G) \cup R(G))$.
An \emph{induced subtrigraph} of a trigraph $G$ is obtained by removing vertices (but no edges) to $G$, analogously to induced subgraphs.
A partial contraction sequence of an $n$-vertex (tri)graph $G$ (to a trigraph $H$) is a sequence of trigraphs $G = G_n, \cdots, G_t = H$ for some $t \in [n]$ such that $G_i$ is obtained from $G_{i+1}$ by performing one contraction.
A~(complete) contraction sequence is such that $t=1$, that is, $H$~is the 1-vertex trigraph.
A~\emph{$d$-sequence} $\mathcal S$ of $G$ is a contraction sequence of $G$ in which the red graph of every trigraph of $\mathcal S$ has maximum degree at most~$d$. 

Assume that there is a partial contraction sequence from a (tri)graph $G$ to a trigraph $H$.
If $u$ is a vertex of $H$, then $u(G) \subseteq V(G)$ denotes the set of vertices eventually contracted into $u$ in $H$.
We denote by $\P(H)$ the partition $\{u(G) : u \in V(H)\}$ of $V(G)$.
If $G$ is clear from the context, we may refer to a \emph{part} of $H$ as any set in $\{u(G) : u \in V(H)\}$.
We will mostly see $d$-sequences as sequences of partitions, that is, $\P_n, \ldots, \P_t$ with $\P_i := \{u(G) : u \in V(G_i)\}$ when $G_n, \ldots, G_t$ is a partial (contraction) $d$-sequence. 



Given a graph $G$ and a partition $\P$ of $V(G)$, the \emph{quotient graph} of $G$ with respect to $\P$ is the graph with vertex set $\P$, where $PP'$ is an edge if there is $u \in P$ and $v \in P'$ such that $uv \in E(G)$.
Given a (tri)graph $G$ and a partition $\P$ of $V(G)$, the \emph{quotient trigraph} $G/\P$ is the trigraph with vertex set $\P$, where $PP'$ is a black edge if for every $u \in P$ and every $v \in P'$, $uv \in E(G)$, and a red edge if either there is $u \in P$ and $v \in P'$ such that $uv \in R(G)$, or there is $u_1, u_2 \in P$ and $v_1, v_2 \in P'$ such that $u_1v_1 \in E(G)$ and $u_2v_2 \notin E(G)$.

A~trigraph $H$ is a \emph{\regu} of another trigraph $G$ if $V(H)=V(G)$, $R(H) \subseteq R(G)$, and $E(G) \subseteq E(H) \subseteq E(G) \cup R(G)$.
That is, $H$ is obtained from $G$ by turning some of its red edges into black edges or non-edges.
We further say that $H$ is \emph{\fullregu} of $G$ if $H$ has no red edge, and thus, is considered as a graph. 
Note that the total graph $\mathcal T(G)$ and the \emph{black graph} $(V(G),E(G))$ of a trigraph $G$ are extreme examples of \fullregus of~$G$.



\subsection{Balanced partition sequences}

The notion of \emph{versatile twin-width} is a crucial opening step to our algorithms; see \cite{twin-width2}.
Let us call \emph{$d$-contraction} a contraction between two trigraphs of maximum red degree at most $d$.
A~\emph{tree of $d$-contractions} of a~trigraph~$G$ (of maximum red degree at most~$d$) is a rooted tree, whose root is labeled by $G$, whose leaves are all labeled by 1-vertex trigraphs $K_1$, and such that one can go from any parent to any of its children by performing a single $d$-contraction.
Observe that $d$-sequences coincide with trees of $d$-contractions that are paths.
A~trigraph $G$ has \emph{versatile twin-width~$d$} if $G$ admits a tree of $d$-contractions in which every internal node, labeled by, say, $F$, has at least $|V(F)|/d$ children each obtained by contracting one of a list of $|V(F)|/d$ pairwise disjoint pairs of vertices of $F$.

It was shown that twin-width and versatile twin-width are functionally equivalent~\cite{twin-width2}.
The relevant consequence for our purposes is that every graph $G$ with a $d'$-sequence admits a~\emph{balanced} $d$-sequence, where $d=h(d')$ depends only on $d'$, i.e., one for which the partitions $\mathcal P_n, \ldots, \mathcal P_1$ are such that for every $i \in [n]$ and $P \in \mathcal P_i$, $|P| \leqslant d \cdot \frac{n}{i}$.
As we will resort to recursion on induced subtrigraphs and quotient trigraphs, we need to keep more information on those subinstances that the mere fact that they have twin-width at most~$d$ (otherwise the twin-width bound could quickly diverge).

This will be done by opening up the proof in~\cite{twin-width2}, and handling divided $0,1,r$-matrices with some specific properties.
Thus we need to recall the relevant definitions.

Given two partitions $\mathcal P, \mathcal P'$ of the same set, we say that $\mathcal P'$ is a \emph{coarsening} of $\mathcal P$ if every part of $\mathcal P$ is contained in a part of $\mathcal P'$, and $\mathcal P, \mathcal P'$ are distinct.
Given a~matrix $M$, we call \emph{row division} (resp.~\emph{column division}) a~partition of the rows (resp.~columns) of~$M$ into parts of consecutive rows (resp.~columns).
A \emph{$(k,\ell)$-division}, or simply \emph{division}, of a matrix $M$ is a pair $(\mathcal R=\{R_1,\dots ,R_k\}, \mathcal C=\{C_1,\dots ,C_\ell\})$ where $\mathcal R$ is a row division and $\mathcal C$ is a column division.
In a matrix division $(\mathcal R,\mathcal C)$, each part $R \in \mathcal R$ is called a \emph{row part}, and each part $C \in \mathcal C$ is called a \emph{column part}.
Given a subset $R$ of rows and a subset $C$ of columns in a matrix $M$, the \emph{zone $M[R,C]$} denotes the submatrix of all entries of $M$ at the intersection between a row of $R$ and a column of $C$.
A~\emph{zone} of a matrix partitioned by $({\mathcal R},{\mathcal C})=(\{R_1, \ldots, R_k\},\{C_1, \ldots, C_\ell\})$ is any $M[R_i,C_j]$ for $i \in [k]$ and $j \in [\ell]$.
A~zone is \emph{constant} if all its entries are identical, \emph{horizontal} if all its columns are equal, and \emph{vertical} if all its rows are equal.
A~\emph{0,1-corner} is a $2 \times 2$ $0,1$-matrix which is not horizontal nor vertical.

Unsurprisingly, \emph{$0,1,r$-matrices} are such that each entry is in $\{0,1,r\}$ where $r$ is an error symbol that should be understood as a~red edge.
A~\emph{neat division} of a~$0,1,r$-matrix is a division for which every zone either contains only $r$~entries or contains no $r$ entry and is horizontal or vertical (or both, i.e., constant).
Zones filled with $r$ entries are called \emph{mixed}. 
A~\emph{neatly divided matrix} is a pair $(M,(\mathcal R,\mathcal C))$ where $M$ is a $0,1,r$-matrix and $(\mathcal R,\mathcal C)$ is a neat division of $M$.
A~\emph{$t$-mixed minor} in a neatly divided matrix is a $(t,t)$-division which coarsens the neat subdivision, and contains in each of its $t^2$ zones at least one mixed zone (i.e., filled with $r$ entries) or a 0,1-corner.
A~neatly divided matrix is said \emph{$t$-mixed free} if it does not admit a~$t$-mixed minor. 

A~\emph{mixed cut of a row part $R \in \mathcal R$} of a~neatly divided matrix $(M,(\mathcal R, \mathcal C=\{C_1,C_2,\ldots\}))$ is an index $i$ such that both $M[R,C_i]$ and $M[R,C_{i+1}]$ are not mixed, and there is a $0,1$-corner in the 2-by-$|R|$ zone defined by the last column of $C_i$, the first column of $C_{i+1}$, and $R$.
The \emph{mixed value of a row part $R \in \mathcal R$} of a~neatly divided matrix $(M,(\mathcal R, \mathcal C=\{C_1,C_2,\ldots\}))$ is the number of mixed zones $M[R,C_j]$ plus the number of mixed cuts between two (adjacent non-mixed) zones $M[R,C_j]$ and $M[R,C_{j+1}]$.
We similarly define the mixed value of a column part $C \in \mathcal C$.
The \emph{mixed value of a neat division} of a $0,1,r$-matrix is the maximum of the mixed values taken over every part.
The \emph{part size} of a division $(\mathcal R,\mathcal C)$ is defined as $\max(\max_{R \in \mathcal R}|R|,\max_{C \in \mathcal C}|C|)$.
A~division is \emph{symmetric} if the largest row index of each row part and the largest column index of each column part define the same set of integers. 
We call \emph{symmetric fusion} of a symmetric division the fusion of two consecutive parts in $\mathcal C$ and of the two corresponding parts in $\mathcal R$.
A~symmetric fusion on a~symmetric division yields another symmetric division.
A~matrix $A := (a_{i,j})_{i,j}$ is said \emph{symmetric} in the usual sense, namely, for every entry $a_{i,j}$ of~$A$, $a_{i,j}=a_{j,i}$.

In what follows, we set $c_d := 8/3(t+1)^22^{4t}$.
The following definition is key.
\begin{definition}
  Let $\mathcal M_{n,d}$ be the class of the neatly divided $n \times n$ symmetric $0,1,r$-matrices $(M,(\mathcal R,\mathcal C))$, such that $(\mathcal R,\mathcal C)$ is symmetric and has:
  \begin{compactitem}
  \item mixed value at most $4c_d$,
  \item part size at most $2^{4c_d+2}$, and
  \item no $d$-mixed minor.
  \end{compactitem}
\end{definition}

The \emph{red number} of a~matrix is the maximum number of $r$ entries in a single column or row of the matrix.

\begin{lemma}\label{lem:ndm-red-number}
  Let $(M,(\mathcal R,\mathcal C)) \in \mathcal M_{n,d}$.
  The red number of $M$ is at most $c_d \cdot 2^{4c_d+4}$.
  Thus, the trigraph whose adjacency matrix is $M$ has maximum red degree at most~$c_d \cdot 2^{4c_d+4}$.
\end{lemma}
\begin{proof}
  Any row or column intersects at most $4c_d$ mixed zones (filled with $r$ entries).
  Each mixed zone has width and length bounded by the part size $2^{4c_d+2}$.
  Hence the maximum total number of $r$ entries on a single row or column is at most $4c_d \cdot 2^{4c_d+2} = c_d \cdot 2^{4c_d+4}$.
\end{proof}

A \emph{coarsening} of a neatly divided matrix $(M,(\mathcal R,\mathcal C))$ is a neatly divided matrix $(M',(\mathcal R',\mathcal C'))$ such that $(\mathcal R',\mathcal C')$ is a coarsening of $(\mathcal R,\mathcal C)$, and $M'$ is obtained from $M$ by setting to $r$ all entries that lie, in $M$ divided by $(\mathcal R',\mathcal C')$, in a zone with at least one $r$ entry or a 0,1-corner.
We also refer to the process of going from $(M,(\mathcal R,\mathcal C))$ to $(M',(\mathcal R',\mathcal C'))$ as \emph{coarsening operation}. 
A coarsening operation from $(M,(\mathcal R,\mathcal C)) \in \mathcal M_{n,d}$ to $(M',(\mathcal R',\mathcal C'))$ is said \emph{invariant-preserving} if $(M',(\mathcal R',\mathcal C')) \in \mathcal M_{n,d}$.

The following lemma is the crucial building block of the current section.

\begin{lemma}[{\cite[Lemma 18]{twin-width2-arxiv}}]\label{lem:coarsening-linear}
  We set $s := 2^{4c_d+4}$.
  Every neatly divided matrix $(M,(\mathcal R,\mathcal C)) \in \mathcal M_{n,d}$ has an invariant-preserving coarsening $(M',(\mathcal R',\mathcal C')) \in \mathcal M_{n,d}$ with $\lfloor n/s \rfloor$ disjoint pairs of identical columns.
  Given $(M,(\mathcal R,\mathcal C))$, both $(M',(\mathcal R',\mathcal C'))$ and the pairs of columns can be computed in $n^{O(1)}$ time.
\end{lemma}

In~\cite{twin-width2-arxiv}, it is not explicitly stated that the invariant-preserving coarsening (hence the pairs of identical columns) can be found in polynomial time.
However it is easy to check that the proof is effective, since it greedily symmetrically fuses two consecutive parts, provided the resulting divided matrix remains in $\mathcal M_{n,d}$.
A~special case of the following observation is shown in \cite[Lemma 19]{twin-width2-arxiv}. 

\begin{lemma}\label{lem:simple-deletion}
  Let $(M,(\mathcal R,\mathcal C)) \in \mathcal M_{n,d}$ be a neatly divided matrix.
  Removing a set of $h$~columns and the $h$~corresponding rows, and possibly removing from the division the parts that are now empty, results in a neatly divided matrix in $\mathcal M_{n-h,d}$.
\end{lemma}
\begin{proof}
  By construction, the new matrix and division are symmetric.
  The new neatly divided matrix remains $d$-mixed free.
  The part size and the mixed value can only decrease.
\end{proof}

\begin{lemma}[{\cite[Beginning of Lemma 20]{twin-width2-arxiv}}]\label{lem:versatile-tww}
  Given any graph $G$ with a~$d$-sequence, one can find in polynomial-time an adjacency matrix $M$ of $G$, such that $(M,(\mathcal R,\mathcal C))$ is a neatly divided matrix of $\mathcal M_{n,2d+2}$ with $(\mathcal R,\mathcal C)$ the finest division of $M$ (i.e., the one where all parts are of size 1).  
\end{lemma}

The adjacency matrix of a trigraph extends the one of a graph by putting $r$ symbols when the vertices of the corresponding row and column are linked by a red edge.
A~neatly divided matrix $(M,(\mathcal R,\mathcal C))$ is said~\emph{conform} to a trigraph $G$ if $M$ is the adjacency matrix of a~trigraph $G'$ such that $G$ is a \regu of $G'$.
Furthermore, we assume (and keep implicit) that we know the one-to-one correspondence between each row (and corresponding column) of~$M$ and vertex of $G$.

\begin{lemma}\label{lem:sequence}
  Let $d$ be a natural, $s := 2^{4c_d+4}$, and $d' := c_d \cdot 2^{4c_d+4}$.
  Let $G$ be an $n$-vertex trigraph given with a~neatly divided matrix $(M,(\mathcal R,\mathcal C)) \in \mathcal M_{n,d}$ conform to~$G$.
  A partial $d'$-sequence $\mathcal S$ from $G$ to a~trigraph $H$ satisfying
  \begin{compactitem}
  \item $|V(H)|=\lfloor \sqrt n \rfloor$, and
  \item $\forall u \in V(H), |u(G)| \leq s \sqrt n$,
  \end{compactitem}
  and a~neatly divided matrix $(M',(\mathcal R',\mathcal C')) \in \mathcal M_{\lfloor \sqrt n \rfloor,d}$ conform to~$H$ can be computed in time $n^{O(1)}$.
\end{lemma}
\begin{proof}
  This is a consequence of~\cref{lem:coarsening-linear,lem:simple-deletion}; see the proof of the more general \cref{lem:seq-to-partial-seq-induction}.
\end{proof}

Combining~\cref{lem:versatile-tww,lem:sequence}, one obtains the following.

\begin{lemma}\label{lem:seq-to-partial-seq}
  Let $d$ be a natural, $s := 2^{4c_d+4}$, and $d' := c_d \cdot 2^{4c_d+4}$.
  Given an $n$-vertex graph $G$ with a $d$-sequence, one can compute in time $n^{O(1)}$ a partition $\P=\{P_1, P_2, \ldots, P_{\lfloor \sqrt n \rfloor}\}$ of $V(G)$ satisfying
  \begin{compactitem}
  \item for every integer $1 \leqslant i \leqslant \lfloor \sqrt n \rfloor$, $|P_i| \leq s \sqrt n \leqslant d' \sqrt n$, and
  \item the red graph of $G/\P$ has maximum degree at most~$d'$.
  \end{compactitem}
\end{lemma}

We will need a stronger inductive form of~\cref{lem:seq-to-partial-seq}, also a~consequence of~\cref{lem:versatile-tww,lem:sequence}. 

\begin{lemma}\label{lem:seq-to-partial-seq-induction}
  Let $\hat{d}$ be a natural, $d=2\hat{d}+2$, and set $s := 2^{4c_d+4}$, and $d' := c_d \cdot 2^{4c_d+4}$.
  Given an $n$-vertex graph $G$ given with a $\hat{d}$-sequence, or an $n$-vertex trigraph $G$ with a~neatly divided matrix $(M,(\mathcal R,\mathcal C)) \in \mathcal M_{n,d}$ such that $M$ is conform to~$G$, one can compute in time $n^{O(1)}$ a~partition $\P=\{P_1, P_2, \ldots, P_{\lfloor \sqrt n \rfloor}\}$ of $V(G)$ with maximum red degree at most~$d'$ satisfying that, for every integer $1 \leqslant i \leqslant \lfloor \sqrt n \rfloor$, $|P_i| \leq s \sqrt n \leq d' \sqrt n$, and for any trigraph $H$ that is 
  \begin{compactitem}
  \item a~\regu of an induced subtrigraph of $G/\P$, or
  \item an induced subtrigraph $G[\bigcup_{i \in J \subseteq [\lfloor \sqrt n \rfloor]} P_i]$,
  \end{compactitem}
  a~neatly divided matrix $(M',(\mathcal R',\mathcal C')) \in \mathcal M_{|V(H)|,d}$ conform to~$H$ can be computed in time $n^{O(1)}$.
\end{lemma}
\begin{proof}
  If we are given a graph $G$ with a~$\hat{d}$-sequence, we immediately compute a~neatly divided matrix $(M,(\mathcal R,\mathcal C)) \in \mathcal M_{n,d}$ conform to~$G$, by~\cref{lem:versatile-tww}.
  We then proceed as if we received the second kind of input.
  
  We will build iteratively the partition $\P=\{P_1, P_2, \ldots, P_{\lfloor \sqrt n \rfloor}\}$ starting from the finest partition.
  At each step we merge two parts, until the number of parts is $\lfloor \sqrt n \rfloor$.
  At this point, we have the desired partition $\P$.

  We iteratively maintain a~trigraph $G^z$ and a~neatly divided matrix $(M^z,(\mathcal R^z,\mathcal C^z)) \in \mathcal M_{n-z+1,d}$ conform to it.
  The maintained partition is just the one corresponding to the parts of $G^z$. 
  Initially, $G^1$ is $G$, and $(M^1,(\mathcal R^1,\mathcal C^1)) = (M,(\mathcal R,\mathcal C)) \in \mathcal M_{n,d}$.
  At step $z$ we do the following.
  We apply \cref{lem:coarsening-linear} on $(M^z,(\mathcal R^z,\mathcal C^z)) \in \mathcal M_{n-z+1,d}$ and obtain, in polynomial-time, an invariant-preserving coarsening $(M'^z,(\mathcal R'^z,\mathcal C'^z)) \in \mathcal M_{n-z+1,d}$, and $h := \lfloor (n-z+1)/s \rfloor$ disjoint pairs of equal columns $\{c_1, c'_1\}, \ldots, \{c_h, c'_h\}$ in $(M'^z,(\mathcal R'^z,\mathcal C'^z))$.
  Let $\{r_1, r'_1\}, \ldots, \{r_h, r'_h\}$ be the corresponding rows, and $\{v_1, v'_1\}, \ldots, \{v_h, v'_h\}$ the corresponding vertices.
  Observe that a coarsening of a neatly divided matrix conform to a trigraph is still conform to that trigraph, since the new matrix may only have some $r$ entries in place of some previously 0 or 1 entries.
  In particular, $(M'^z,(\mathcal R'^z,\mathcal C'^z))$ is conform to~$G^z$.
  
  There is at least one pair $\{v_i, v'_i\}$ whose contraction forms a part of size at most~$n/h$.
  Indeed, otherwise the union of the parts corresponding to $v_1, v'_1\, \ldots, v_h, v'_h$ is larger than~$n$.
  We remove $c'_i$ and $r'_i$ from $(M'^z,(\mathcal R'^z,\mathcal C'^z))$.
  By~\cref{lem:simple-deletion}, we obtain a neatly divided matrix of $M_{n-z,d}$ that we denote by $(M^{z+1},(\mathcal R^{z+1},\mathcal C^{z+1}))$.
  As we stop when $n-z+1=\lfloor \sqrt n \rfloor$, it means that the maximum size of a part of our partition is at most $n/h \leqslant sn/\sqrt n=s \sqrt n$. 
  The bound on the maximum red degree of the obtained partition (actually of all maintained partitions) is given by~\cref{lem:ndm-red-number}.
  
  We now show to find, for any \regu $H$ of an induced subtrigraph of $G/\P$, a~neatly divided matrix $(M',(\mathcal R',\mathcal C')) \in \mathcal M_{|V(H)|,d}$ conform to~$H$.
  We first observe, as a~consequence of \cref{lem:coarsening-linear,lem:simple-deletion}, that $(M^{\lfloor \sqrt n \rfloor},(\mathcal R^{\lfloor \sqrt n \rfloor},\mathcal C^{\lfloor \sqrt n \rfloor})) \in \mathcal M_{{\lfloor \sqrt n \rfloor},d}$ is conform to $G/\P$.
  Taking an induced subgraph $H'$ of $G/\P$ (i.e., removing vertices from it), we get, by removing the corresponding rows and columns in $(M^{\lfloor \sqrt n \rfloor},(\mathcal R^{\lfloor \sqrt n \rfloor},\mathcal C^{\lfloor \sqrt n \rfloor}))$ a neatly divided matrix $(M',(\mathcal R',\mathcal C')) \in \mathcal M_{|V(H')|,d}$ conform to $H'$, by~\cref{lem:simple-deletion}.
  Note finally that taking a \regu $H$ of $H'$, we can simply keep $(M',(\mathcal R',\mathcal C'))$ as a neatly divided matrix of $\mathcal M_{|V(H)|,d}$ conform to~$G$.
  The second item, concerning induced subtrigraphs $G[\bigcup_{i \in J \subseteq [\lfloor \sqrt n \rfloor]} P_i]$ is a simple application of~\cref{lem:simple-deletion}, and works more generally for any induced subtrigraph of $G$.
\end{proof}

In effect, we will only apply \cref{lem:seq-to-partial-seq-induction} for graphs $G$ and $H$, i.e., when $H$ is an induced subgraph of $G$ or a~\fullregu of an induced subtrigraph of $G/\P$.
Indeed, the structures $H$ will correspond to subinstances.
We want those to be graphs, so that the tackled graph problem is well-defined on them.


\section{Approximation algorithms for \mis}

We naturally start our study with \mis, a central problem that is very inapproximable~\cite{Hastad96,Zuckerman07}, and yet constitutes the textbook example of our approach. 

\subsection{Subexponential-time constant-approximation algorithm}

We present a subexponential-time $O_d(1)$-approximation for \swis on graphs given with a~$d$-sequence, which we recall, is unlikely to exist in general graphs~\cite{Chalermsook13}. 

\begin{lemma}\label{lem:mis-subexp-approx}
  Let $d'$ be a natural, $s := 2^{4c_{d'}+4}$, and $d := c_{d'} \cdot 2^{4c_{d'}+4}$.
  Assume $n$-vertex inputs $G$, vertex-weighted by $w$, are given with a~$d'$-sequence.
  \wis can be $(d+1)$-approximated in time $2^{O_d(\sqrt n)}$ on these inputs.
\end{lemma}

\begin{proof}

  By~\cref{lem:seq-to-partial-seq}, we compute in polynomial time a partition $\P=\{P_1, \ldots, P_{\lfloor \sqrt n \rfloor}\}$ of~$V(G)$ whose parts have size at most $s \sqrt n$ and such that $\mathcal R(G/\P)$ has maximum degree at most~$d$.

  For every integer $1 \leqslant i \leqslant \lfloor \sqrt n \rfloor$, we compute a heaviest independent set in $G[P_i]$, say~$S_i$.
  Even with an exhaustive algorithm, this takes time $\sqrt n \cdot s^2 \sqrt n \cdot 2^{s \sqrt n}=2^{O_d(\sqrt n)}$.
  We then $(d+1)$-color (in linear time) $\mathcal R(G/\P)$, which is possible since this graph has maximum degree at most~$d$.
 This defines a coarsening of $\P$ in $d+1$ parts $\mathcal Q = \{C_1, \ldots, C_{d+1}\}$.
 Thus, $\mathcal Q$ is a partition of $V(G)$ such that $C_j$ consists of all the parts $P_i \in \P$ receiving color $j$ in the $(d+1)$-coloring of $\mathcal R(G/\P)$.

 For every $j \in [d+1]$, let $H_j$ be the \emph{graph} $(G/\P)[C_j]$\footnote{We use this notation as a slight abuse of notation for $(G/\P)[\{P_i~:~P_i \subseteq C_j\}]$.} vertex-weighted by $P_i \subseteq C_j \mapsto w(S_i)$.
 Note that $(G/\P)[C_j]$ can indeed be assimilated to a graph, since it has, by design, no red edge.
 We compute a heaviest independent set in $H_j$, say $R_j$.
 This takes time $(d+1) \cdot n \cdot 2^{\sqrt n}=2^{O_d(\sqrt n)}$.
 We output $\bigcup_{P_i \subseteq R_j} S_i$ for the index $j \in [d+1]$ maximizing $\sum\limits_{P_i \subseteq R_j} w(S_i)$.

 This finishes the description of the algorithm.
 We already argued that its running time is $2^{O_d(\sqrt n)}$.
 We shall justify that it does output an independent set of weight at least a $\frac{1}{d+1}$ fraction of the optimum $\alpha(G)$.

 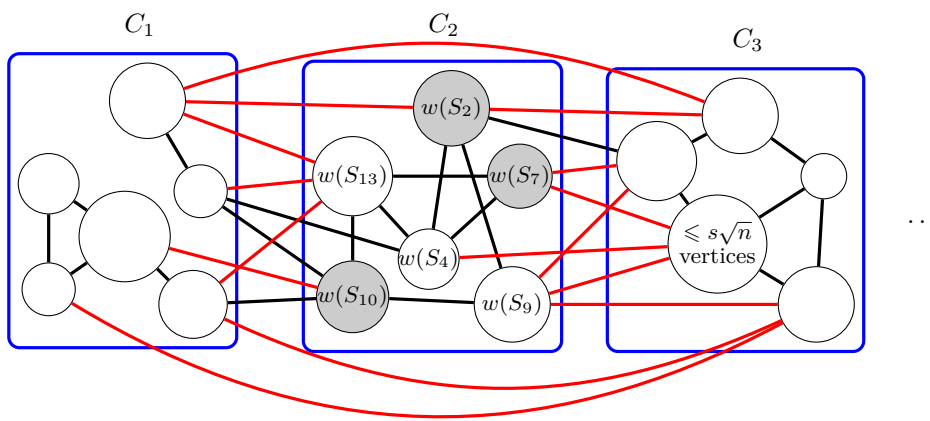
\begin{figure}[h!]
   \centering
   \begin{tikzpicture}
     \foreach \i/\j/\l/\s in {0/0/p1/0.8, 0.3/2/p2/1, 1.2/1.1/p3/0.85, -1/1.1/p4/1.05, -1/-0.5/p5/0.95, 1.1/-0.6/p6/1,
       -4/0.3/q1/1.2, -3.7/2.1/q2/1, -3/0.9/q3/0.7, -5/1/q4/0.8, -5/-0.4/q5/0.7, -3.1/-0.6/q6/0.9,
        4.1/1.9/r2/1, 5.2/1.1/r3/0.6, 3/1.3/r4/1.05, 3.8/0.2/r5/1.3, 5.1/-0.6/r6/1}{
       \node[draw,circle,minimum size=\s cm] (\l) at (\i,\j) {} ;
     }
     \foreach \i/\j/\l/\s in {0.3/2/p2/1, 1.2/1.1/p3/0.85, -1/-0.5/p5/0.95}{
       \node[circle,fill,fill opacity=0.2,minimum size=\s cm] (\l) at (\i,\j) {} ;
     }
      \foreach \i/\j/\l in {0/0/{w(S_4)}, 0.3/2/{w(S_2)}, 1.2/1.1/{w(S_7)}, -1/1.1/{w(S_{13})}, -1/-0.5/{w(S_{10})}, 1.1/-0.6/{w(S_9)}}{
       \node at (\i,\j) {\footnotesize{$\l$}} ;
     }
      \node at (3.8,0.35) {\footnotesize{$\leqslant s \sqrt n$}} ;
      \node at (3.8,0.05) {\footnotesize{vertices}} ;

    \node at (6.5,0.5) {$\ldots$} ;
     
     \node[draw,very thick,blue,rounded corners,fit=(p1) (p2) (p3) (p4) (p5) (p6)] {} ;
     \node[draw,very thick,blue,rounded corners,fit=(q1) (q2) (q3) (q4) (q5) (q6)] {} ;
     \node[draw,very thick,blue,rounded corners,fit=(r2) (r3) (r4) (r5) (r6)] {} ;
     \foreach \i/\j/\k in {-3.8/3.1/{C_1},0.2/3.1/{C_2},4.2/2.9/{C_3}}{
       \node at (\i,\j) {$\k$} ;
     }
     
     \foreach \i/\j in {p1/p2,p1/p3,p1/p4,p3/p4,p4/p5,p5/p6,p6/p2,
       q2/q3,q1/q6,q1/q4,q4/q5,q5/q1,
       r2/r3,r3/r6,r6/r5,r5/r4,r4/r2,r5/r3,
       q3/p1,q3/p5,q6/p5,p2/r4}{
       \draw[very thick] (\i) -- (\j) ;
     }

     \foreach \i/\j in {q3/p4,q6/p4,q2/p4,q2/p2,q1/p5,p3/r4,p3/r5,p6/r4,p6/r5,p6/r6,p1/r5,p2/r2}{
       \draw[very thick,red] (\i) -- (\j) ;
     }
     \foreach \i/\j/\b in {q2/r2/20,q5/r6/-30,q6/r6/-25}{
       \draw[very thick,red] (\i) to [bend left=\b] (\j) ;
     }
   \end{tikzpicture}
   \caption{The trigraph $G/\P$ with its $\lfloor \sqrt n \rfloor$ vertices, each corresponding to a subset of at most $s \sqrt n$ vertices of $G$.
     The weights $w(S_i)$ of heaviest independent sets $S_i$ of $G[P_i]$ for each part $P_i$ of the color class $C_2$ of the $d+1$-coloring of $\mathcal R(G/\P)$.
     A heaviest independent set in the so-weighted $(G/\P)[C_2]$ (shaded) corresponds to an optimum solution in $G[\bigcup_{P_i \subseteq C_2} P_i]$.
   One of these $d+1$ independent sets is a $d+1$-approximation.}
   \label{fig:alg-mis}
 \end{figure}

 \medskip

 \textbf{$I$ is indeed an independent set.}
 For any $j \in [d+1]$, consider two vertices $x, y \in \bigcup_{P_i \subseteq R_j} S_i$.
 If $\{x,y\} \in S_i$ for some $i$, then $x$ and $y$ are non-adjacent since $S_i$ is an independent set of $G[P_i]$.
 Else $x \in S_i$ and $y \in S_{i'}$ for some $i \neq i'$.
 $P_i$ and $P_{i'}$ are not linked by a black edge in $(G/\P)[C_j]$ since $R_j$ is an independent set in $H_j$, nor they can be linked by a red edge (there are none in $(G/\P)[C_j]$).
 Thus again, $x$ and $y$ are non-adjacent in $G$.

 \medskip

 \textbf{$I$ has weight at least $\frac{\alpha(G)}{d+1}$}.
 We claim that $\bigcup_{P_i \subseteq R_j} S_i$ is a heaviest independent set of $G[C_j]$.
 Note that the $P_i$s that are included in $C_j$ (and partition it) form a module partition of $G[C_j]$.
 In particular, any heaviest independent set intersecting some $P_i \subseteq C_j$ has to contain a heaviest independent of $G[P_i]$.
 This is precisely what the algorithm computes.
 Then a heaviest independent set in $G[C_j]$ packs such subsolutions to maximize the total weight, which is what is computed in $H_j$. 

 We conclude by the pigeonhole principle, since a heaviest independent set $X$ of $G$ is such that $w(X \cap C_j) \geqslant \frac{\alpha(G)}{d+1}$ for some $j \in [d+1]$.
\end{proof} 

\subsection{Improving the approximation factor}

We notice in this short section that the approximation factor of~\cref{lem:mis-subexp-approx} can be improved using the notion of \emph{clustered coloring}.
The \emph{clustered chromatic number} of a class of graphs is the smallest integer $k$ such that there is a constant $c$ for which all the graphs of the class can be $k$-colored such that every color class induces a subgraph whose connected components have size at most $c$.
A proper coloring is a particular case of clustered coloring when $c=1$.

Instead of properly coloring the red graph, as we did in the proof of~\cref{lem:mis-subexp-approx}, we could use less colors and allow for small monochromatic components (in place of monochromatic components of size~1).
We use for that the following bound due to Alon et al.

\begin{theorem}[\cite{Alon03}]\label{lem:clustered-coloring}
  The class of graphs of maximum degree at most~$d$ has clustered chromatic number at most~$\lceil \frac{d+2}{3} \rceil$.
\end{theorem} 

We can use this lemma to improve our approximation algorithms.

\begin{theorem}\label{thm:improved-mis}
  On inputs as in~\cref{lem:mis-subexp-approx} with $s := 2^{4c_{d'}+4}$, and $d := c_{d'} \cdot 2^{4c_{d'}+4}$, \wis further admits an  $\lceil \frac{d+2}{3} \rceil$-approximation algorithm in time $2^{O_d(\sqrt n)}$.
\end{theorem}

\begin{proof}
   Again, we compute in polynomial time a partition $\P=\{P_1, \ldots, P_{\lfloor \sqrt n \rfloor}\}$ of~$V(G)$ whose parts have size at most $s \sqrt n$ and such that $\mathcal R(G/\P)$ has maximum degree at most~$d$, using~\cref{lem:seq-to-partial-seq}.
   Let $c$ be the constant such that $\mathcal R(G/\P)$ admits a clustered coloring using $\lceil \frac{d+2}{3} \rceil$ colors such that each color class $C_j$ (with $j \in [\lceil \frac{d+2}{3} \rceil]$) is such that the connected components $C_j^1, C_j^2, \ldots, C_j^{h_j} \subseteq C_j$ of $\mathcal R(G/\P)[C_j]$ have size at most $c$ each. 
   This coloring is guaranteed to \emph{exist} by~\cref{lem:clustered-coloring}.
   Due to the overall running time, we might as well compute it by exhaustive search, in time $2^{O_d(\sqrt n)}$.

   For every $j \in \lceil \frac{d+2}{3} \rceil$ and $h \in [h_j]$, we denote by $P_1(C_j^h), \ldots, P_{c(j,h)}(C_j^h)$ the $c(j,h) \leqslant c$ parts $P_i \in \P$ that are included in $C_j^h$.
   For every $j \in \lceil \frac{d+2}{3} \rceil$, every $h \in [h_j]$, and every $J \subseteq [c(j,h)]$, we compute a heaviest independent set in $G[\bigcup_{z \in J}  P_z(C_j^h)]$, which we denote by $S_{j,h,J}$.
   This takes time $O(\sqrt n \cdot 2^c \cdot 2^{sc \sqrt n}) = 2^{O_d(\sqrt n)}$ since $|\bigcup_{z \in J}  P_z(C_j^h)| \leqslant c \cdot s \sqrt n$.

   For each $C_j$, in time $(2^c)^{\sqrt n}=2^{c \sqrt n}$, we exhaustively try all subsets $X \subseteq \bigcup_{P_i \in C_j} P_i$ that are unions of $S_{j,h,J}$ filtering them out when $G[X]$ is not edgeless, and keep a heaviest of them, say $R_j$.
   Since there can only be black edges or non-edges between some $P_i \in C_j^h$ and $P_{i'} \in C_j^{h'}$ with $h \neq h'$, it is clear that a heaviest independent set of $G[\bigcup_{P_i \in C_j} P_i]$ is indeed a union of $S_{j,h,J}$ (with fixed $j$).
   We output a heaviest set among the $R_j$s, which is the desired $\lceil \frac{d+2}{3} \rceil$-approximation.
   The running time is as claimed.
\end{proof}

\subsection{Time-approximation trade-offs}\label{sec:time-app-tradeoffs}

\Cref{lem:mis-subexp-approx,thm:improved-mis} run exhaustive algorithms on induced subgraphs of size $O_d(\sqrt n)$.
As such, the latter inputs keep the same twin-width upper bound.
To speed up the algorithm (admittedly while worsening the approximation factor) it is tempting to recursively call our very algorithm.
We show that this leads to a time-approximation trade-off parameterized by an integer $q=0, \ldots, O_d(\log \log n)$.
At one end of this discrete curve, one finds the exact exponential algorithm ($q=0$), and more interestingly the $d+1$-approximation in time $2^{O_d(\sqrt n)}$ ($q=1$), while at the other end lies a polynomial-time algorithm with approximation factor~$n^\varepsilon$, where $\varepsilon > 0$ can be made as small as desired.

As we will deal with the same kind of recursions for several problems, we show the following generic abstraction.

\begin{lemma}\label{lem:induction-approx-fast}
  Let $\hat{d}$ be a natural, $d'=2\hat{d}+2$, and $d := c_{d'} \cdot 2^{4c_{d'}+4}$.
  Let $\Pi$ be an optimization graph problem where inputs come with a~\mbox{$\hat{d}$-sequence} of their $n$-vertex graph $G$, or with a~neatly divided matrix $(M,(\mathcal R,\mathcal C)) \in \mathcal M_{n,d'}$ conform to~$G$.
  Let $\P$ be the partition of~$V(G)$ given by~\cref{lem:seq-to-partial-seq-induction}.
  Assume that
  \begin{compactenum}
  \item\label{it:exact} $\Pi$ can be exactly solved in time $2^{O(n)}$, and there are constants $c_1, c_2, c_3$, and a function $f \geqslant 1$ such that 
  \item\label{it:induct} a~\mbox{$d^{c_3}r^2$-approximation} of $\Pi$ on $G$ can be built in time $n^{c_2}$ by using at most $n^{c_1}$ calls to an $r$-approximation of $\Pi$ --or another optimization problem $\Pi'$ already satisfying the conclusion of the lemma-- on an induced subgraph of~$G$ with at most $f(d) \sqrt n$ vertices or a~\fullregu of an induced subtrigraph of $G/\P$ (on at most $\sqrt n$ vertices).
  \end{compactenum}
  Then $\Pi$ can be $d^{c_3 (2^q-1)}$-approximated in time $$(f(d)^q n)^{(2-2^{-q})(c_1+c_2)} \cdot 2^{f(d)^{2(1-2^{-q})} n^{2^{-q}}},$$ for any non-negative integer~$q$.
\end{lemma}
\begin{proof}
  The proof is by induction on $q$.
  The case $q=0$ is implied by~\cref{it:exact}.
  The case $q=1$, and the induction step in general, is nothing more than an abstraction of~\cref{lem:mis-subexp-approx}, where exhaustive algorithms are replaced by recursive calls.
  
  For any $q \geqslant 0$, we assume that $\Pi$ can $d^{c_3 (2^q-1)}$-approximated in the claimed running time, and show the same statement for the value $q+1$.
  Following~\cref{it:induct}, we run this algorithm --or one for another optimization problem $\Pi'$ satisfying the conclusion of the lemma-- at most $n^{c_1}$ times on $f(d) \sqrt n$-vertex induced subgraphs of the input graph $G$ or on \fullregus of induced subtrigraphs of $G/\P$.
  The latter graphs have at most $\sqrt n \leqslant f(d) \sqrt n$ vertices.
  By~\cref{lem:seq-to-partial-seq-induction}, we can compute in polynomial time a~neatly divided matrix $(M',(\mathcal R',\mathcal C')) \in \mathcal M_{|V(H)|,d'}$ conform to~$H$, for each graph $H$ of a~recursive call; hence the induction applies.   
  
  Overall this takes time at most $$n^{c_1}+n^{c_2} \cdot \left( (f(d)^q \cdot f(d) \sqrt n)^{(2-2^{-q})(c_1+c_2)} \cdot 2^{f(d)^{2(1-2^{-q})} (f(d) \sqrt n)^{2^{-q}}} \right)$$
  $$\leqslant (f(d)^{q+1} n)^{c_1+c_2 + \frac{1}{2} (2-2^{-q})(c_1+c_2)} \cdot 2^{f(d)^{2(1-2^{-q})+2^{-q}} n^{\frac{2^{-q}}{2}}}$$
  $$= (f(d)^{q+1} n)^{(2-\frac{2^{-q}}{2})(c_1+c_2)} \cdot 2^{f(d)^{2-2^{-q+1}+2^{-q}} n^{2^{-(q+1)}}}$$
  $$= (f(d)^{q+1} n)^{(2-2^{-(q+1)})(c_1+c_2)} \cdot 2^{f(d)^{2(1-2^{-(q+1)})} n^{2^{-(q+1)}}}.$$

  For the first inequality, we assume that the two summands are larger than 2, so their sum can be bounded by their product.

   Besides we get an approximation of factor at most $(d^{c_3 (2^q-1)})^2 d^{c_3} = d^{c_3 (2^{q+1}-1)}$.
\end{proof}

In more legible terms we have proved that:
\begin{lemma}\label{lem:induction-approx-leg-fast}
  Problems $\Pi$ satisfying the assumptions of~\cref{lem:induction-approx-fast} can be $d^{O(1)(2^q-1)}$-approximated in time $2^{O_{d,q}(\sqrt[2^q]{n})}$, for any non-negative integer $q$.
\end{lemma}

If most graph problems admit single-exponential algorithms, we will deal with such a~problem that is only known to be solvable in time $2^{O(n \log n)}$.
Therefore we prove a~variant of~\cref{lem:induction-approx-fast} with a~slightly worse running time.

\begin{lemma}\label{lem:induction-approx}
  Let $\Pi$ be solvable in time $2^{O(n \log n)}$ and satisfy the second item of~\cref{lem:induction-approx-fast}.
  Then $\Pi$ can be $d^{c_3 (2^q-1)}$-approximated in time $$2^{\left((c_1+c_2) (2-2^{-q}) \log f(d) + f(d)^{2(1-2^{-q})} n^{2^{-q}}\right) \log n},$$ for any non-negative integer $q$.
\end{lemma}

\begin{proof}
  We follow the proof of~\cref{lem:induction-approx-fast} when the induction now gives a running time of $$n^{c_2}+n^{c_1} \cdot 2^{\left((c_1+c_2) (2-2^{-q}) \log f(d) + (f(d) \sqrt n)^{2^{-q}}\right) \log (f(d) \sqrt n)}$$
  $$\leqslant 2^{\left((c_1+c_2) (2-2^{-(q+1)}) \log f(d) + f(d)^{2(1-2^{-(q+1)})} n^{2^{-(q+1)}}\right) \log n}.$$
\end{proof}

Again the previous lemma can be rewritten as:
\begin{lemma}\label{lem:induction-approx-leg}
  Problems $\Pi$ satisfying the assumptions of~\cref{lem:induction-approx} can be $d^{O(1)(2^q-1))}$-approximated in time $2^{O_{d,q}(\sqrt[2^q]{n} \log n)}$, for any non-negative integer $q$.
\end{lemma}

We derive from~\cref{lem:induction-approx} the following notable regimes.

\begin{theorem}\label{lem:polytime-approx}
  Problems $\Pi$ satisfying the assumptions of~\cref{lem:induction-approx} admit polynomial-time $n^\varepsilon$-approximation algorithms, for any $\varepsilon > 0$.
\end{theorem}

\begin{proof}
  This is the particular case $q = \lceil \log \frac{\varepsilon \log n}{2c_3 \log d}\rceil$.

  Indeed the approximation factor is then at most $d^{c_3 (2^q-1)} \leqslant d^{2c_3 \frac{\varepsilon \log n}{2c_3 \log d}} = 2^{\varepsilon \log n}=n^\varepsilon$, while the running time is at most $$2^{\left((c_1+c_2) (2-2^{-q}) \log f(d) + f(d)^{2(1-2^{-q})} n^{2^{-q}} \right) \log n} \leqslant 2^{\left(2(c_1+c_2)\log f(d) + f(d)^2 n^{\frac{2c_3 \log d}{\varepsilon \log n}}\right)\log n}$$
  $$=n^{2(c_1+c_2) \log f(d) + f(d)^2 d^{\frac{2c_3}{\varepsilon}}}.$$

  If further $\Pi$ can be solved exactly in time $2^{O(n)}$ (hence satisfies the assumptions of~\cref{lem:induction-approx-fast}), one obtains a better running time, where the exponent of $n$ does not depend on $\varepsilon$.
  Indeed,
  \[(f(d)^q n)^{(2-2^{-q})(c_1+c_2)} 2^{f(d)^{2(1-2^{-q})} n^{2^{-q}}} \leqslant \left(\frac{\varepsilon \log n}{c_3 \log d}\right)^{2(c_1+c_2) \log f(d)} 2^{f(d)^2 d^{\frac{2c_3}{\varepsilon}}} n^{2(c_1+c_2)}.\qedhere\]
\end{proof}

\begin{theorem}\label{lem:log-approx}
  Problems $\Pi$ satisfying the assumptions of~\cref{lem:induction-approx-fast}, resp.~\cref{lem:induction-approx}, admit a $\log n$-approximation algorithm running in time $2^{O_d(n^{\frac{1}{\log \log n}})}$, resp.~$2^{O_d(n^{\frac{1}{\log \log n}} \log n)}$.
\end{theorem}
\begin{proof}
  This is the particular case $q = \lfloor \log \left( \frac{\log \log n}{c_3 \log d} + 1 \right)\rfloor$.

  This value is computed such that the approximation factor $d^{c_3(2^q-1)}$ is at most $\log n$.
  It can be easily checked that the running times are as announced.
\end{proof}

We derive the following for \wis.

\begin{theorem}\label{lem:mis-iteration}
  \wis on $n$-vertex graphs $G$ (vertex-weighted by $w$) given with a~$d'$-sequence satisfies the assumptions of~\cref{lem:induction-approx-fast}.
  In particular, this problem admits
  \begin{compactitem}
  \item a $(d+1)^{2^q-1}$-approximation in time $2^{O_{d,q}(n^{2^{-q}})}$, for every integer $q \geqslant 0$,
  \item an $n^\varepsilon$-approximation in polynomial-time $O_{d,\varepsilon}(1) \log^{O_d(1)}n \cdot n^{O(1)}$, for any $\varepsilon > 0$, and
  \item a $\log n$-approximation in time $2^{O_d(n^{\frac{1}{\log \log n}})}$,
  \end{compactitem}
  with $d := c_{2d'+2} \cdot 2^{4c_{2d'+2}+4}$. 
\end{theorem}
\begin{proof}
  Even the exhaustive algorithm exactly solves \swis in time $2^{O(n)}$.
  We thus focus on showing that \swis satisfies the second item of~\cref{lem:induction-approx-fast}.
  We set $c_1 \geqslant 1$ as the required exponent to turn a $d'$-sequence into a~neatly divided matrix of $\mathcal M_{n,2d'+2}$ conform to~$G$, $c_2=\frac{1}{2}+\eta$ for any fixed $\eta>0$, the appropriate $1 < c_3 \leqslant 2$, and $f(d)=d \geqslant 1$.

  The algorithm witnessing the second item is simply the proof of~\cref{lem:mis-subexp-approx}.
  We first check that this algorithm makes $\lfloor \sqrt n \rfloor + d + 1$ recursive calls on induced subgraphs of the input~$G$: each of the $\lfloor \sqrt n \rfloor$ graphs $G[P_i]$ where $P_i$ has indeed size at most $O_d(\sqrt n)$, and each of the $d+1$ graphs $(G/\P)[C_j]$ (indeed an induced subgraph of $G$ by definition of the black graph of a~trigraph) on at most $\sqrt n$ vertices.

  We finally assume that each recursive call outputs an $r$-approximation of \swis.
  Let $j \in [d+1]$ be such that $w(C_j \cap I) \geqslant \frac{1}{d+1} w(I)$ for $I$ a heaviest independent set of $G$ vertex-weighted by $w$.
  Let $J \subseteq [\lfloor \sqrt n \rfloor]$ be the indices of the $P_i$s that are intersected by $C_j \cap I$, that is, $J = \{i~:~P_i \cap (C_j \cap I) \neq \emptyset\}$.
  For every $i \in J$, set $w_i = w(P_i \cap I)$.
  Each recursive call on some $P_i$ with $i \in J$, yields an independent set of weight at least $\frac{w_i}{r}$, by assumption.
  Thus the weights that our algorithm puts on $(G/\P)[C_j]$ are such that it has an independent set of weight at least $\Sigma_{i \in J} \frac{w_i}{r}=\frac{w(C_j \cap I)}{r}$.
  As we run an $r$-approximation on this graph, we get an independent set of weight at least $\frac{w(C_j \cap I)}{r^2} \geqslant \frac{w(I)}{(d+1)r^2}$.
  Thus \swis satisfies the assumptions of~\cref{lem:induction-approx-fast}, and we conclude.
\end{proof}

\section{Finding the suitable generalization: the case of \coloring} \label{subsec:coloring}

In this section, we deal with the \coloring problem.
Unlike for \swis, we cannot solely resort to recursively calling our \coloring algorithm on smaller graphs.
The right problem generalization needs to be found for the inductive calls to work through, and it happens to be \setcoloring.

In the \setcoloring problem, the input is a couple $(G, b)$ where $G$ is a graph, and $b$ is a function assigning a positive integer to each vertex of $G$. The goal is to find, for each $v \in V(G)$, a set $S_v$ of at least $b(v)$ colors such that $S_u \cap S_v = \emptyset$ whenever $uv \in E(G)$, and minimizing $|\cup_{v \in V(G)} S_v |$. Let $\chi_b(G)$ be the optimal value of \setcoloring for $(G, b)$. Observe that \coloring corresponds to the case where $b(v)=1$ for every $v \in V(G)$.

\begin{theorem}\label{thm:coloring-app}
\setcoloring (and hence \coloring) on $n$-vertex graphs $G$ given with a $d'$-sequence satisfies the assumptions of~\cref{lem:induction-approx}. In particular, this problem admits
\begin{compactitem}
	\item a $(d+1)^{2^q-1}$-approximation in time $2^{O_{d, q}(n^{2^{-q}} \log n)}$, for every integer $q \geqslant 0$, and 
	\item an $n^{\varepsilon}$-approximation in polynomial-time for any $\varepsilon > 0$.
\end{compactitem}
with $d := c_{2d'+2} \cdot 2^{4c_{2d'+2}+4}$.
\end{theorem}
\begin{proof}
It is known~\cite{Ned08} that \setcoloring can be solved using the inclusion-exclusion principle in time $O^*(\max_{v \in V(G)}b(v)^n) = 2^{O(n \log n)}$.
We now prove that it satisfies the second item of~\cref{lem:induction-approx-fast}. We denote by $\mathcal{A}$ the $r$-approximation algorithm of the statement, which we will use on instances of \setcoloring. In particular, we will call it at most $\sqrt{n} + 1$ times, and will obtain at the end a $(d+1)r^2$-approximation on our input $(G, b)$ in polynomial time.

We first apply~\cref{lem:seq-to-partial-seq-induction} to get, in polynomial-time, a partition $\P = \{P_1, \dots, P_{\lfloor \sqrt{n}\rfloor}\}$ of $V(G)$ whose parts have size at most $d \sqrt{n}$ and such that $R(G/\P)$ has maximum degree at most $d$. 
For every $i \in [\lfloor \sqrt{n} \rfloor]$, we use $\mathcal{A}$ to compute an $r$-approximated solution $c_{P_i}$ of $(G[P_i], b_{|P_i})$. 
We denote by $b'$ the function which assigns, to each $P_i$, the number of colors of $c_{P_i}$. 
We now compute, in polynomial-time, a proper $(d+1)$-coloring of $R(G/\P)$, which defines the sets $C_1$, $\dots$, $C_{d+1}$. For each $j \in [d+1]$, we construct another \setcoloring instance consisting of the graph $H_j = (G/\P)[C_j]$ (recall that this trigraph has no red edge, and can thus be seen as a graph), together with the function $b'_{|C_j}$. Again we use $\mathcal{A}$ to compute an $r$-approximated solution on $(H_j,  b'_{|C_j})$. We denote by $c_H$ this solution.
Let $G_j$ be the subgraph of $G$ induced by $\cup_{P_i \in C_j} P_i$, and $b_j$ the restriction of $b$ to $V(G_j)$. We now show how to construct a solution $c_j$ of \setcoloring to $(G_j, b_j)$ from $c_H$ and all $c_{P_i}$.
Recall that for every $P_i \in C_j$, every $v \in P_i$, we have that $c_{P_i}(v)$ is a subset of $\{1, \dots, b'(P_i)\}$ of size at least $b(v)$, and that $c_H(P_i)$ is a subset of size at least $b'(P_i)$. Hence, for each $P_i \in C_j$, one can choose an arbitrary bijection $\tau$ from $\{1, \dots, b'(P_i)\}$ to $c_H(P_i)$, and define to each vertex $v \in P_i$ the set $c_j(v)$ as $\{\tau(x) : x \in c_{P_i}(v)\}$.

By construction, this solution is a feasible one for the instance $(G_j, b_j)$.
Let us prove that it is an $r^2$-approximation of $\chi_{b_j}(G_j)$.
First, by definition of $c_H$, our solution uses at most $r \cdot \chi_{b'_{|C_j}}(H_j)$ colors. Then, by definition of $c_{P_i}$ for every $P_i \in C_j$, we have $b'_{C_j}(P_i) \leqslant r \cdot \chi_{b_{|P_i}}(G[P_i])$.
Now, denote by $\Gamma$ the function which assigns to each $P_i \in C_j$ the number $\chi_{b_{|P_i}}(G[P_i])$. We now use the following claim, whose proof is left to the reader.
\begin{claim}\label{claim:larger-set}
Let $(G, b)$ be an instance of \setcoloring, and $r \in \R_+$. It holds that $\chi_{r \cdot b}(G) \leqslant r \cdot \chi_b(G)$, where $r \cdot b$ is the function which assigns $r \cdot b(v)$ to each $v \in V(G)$.
\end{claim}

This implies $\chi_{b'_{|C_j}}(H_j) \leqslant r \cdot \chi_{\Gamma}(H_j)$, and thus our solution uses at most $r^2 \cdot \chi_{\Gamma}(H_j)$ colors. We now prove the following claim.

\begin{claim}
$\chi_{\Gamma}(H_j) \leqslant \chi_{b_j}(G_j)$.
\end{claim}
\begin{proof}[Proof of the claim]
Let $c$ be an optimal solution for $(G_j, b_j)$. For every distinct $P_i$, $P_{i'} \in C_j$ such that $P_iP_{i'}$ is an edge of $H_j$, it holds that there are all possible edges between $P_i$ and $P_{i'}$ in $G_j$ (by definition of the coloring $C_1$, $\cdots$, $C_{d+1}$), hence it holds that $\bigcup_{v \in P_i} c(v)$ and $\bigcup_{v \in P_{i'}}c(v)$ have empty intersection. Moreover, by definition of $\Gamma$, we have that $\bigcup_{v \in P_i} c(v)$ is of size at least $\Gamma(P_i)$, hence the function which assigns $\bigcup_{v \in P_i} c(v)$ to each $P_i$ is a feasible solution for $(H_j, \Gamma)$ using at most $\chi_{b_j}(G_j)$ colors.
\end{proof}

We now have in hand an $r^2$-approximated solution of $(G_j, b_j)$ for every $j \in [d+1]$, which can be turned into a $(d+1)r^2$-approximated solution of $(G, b)$, as desired.
\end{proof}

\section{Edge-based problems: the case of \mim} \label{subsec:matching}

So far, we only considered problems where approximated solutions in each part $P_i$ of a~partition $\P$ of $V(G)$ of small width, and in some selected induced subgraphs of $(V(G/\P),E(G/\P))$, were enough to build an approximated solution for $G$.\footnote{The improvement based on clustered coloring slightly departed from that simple scheme.}
We now handle problems for which a number of edges is to be optimized.
Now all competitive solutions can integrally lie in between pairs of parts $P_i, P_j$ linked by a black or a red edge in $G/\P$.
This complicates matters, and forces us to be competitive there as well, naturally splitting the algorithm into three subroutines.

We present the algorithms for \msim where one is given, in addition to the input graph $G$ (possibly with edge weights), a subset $Y \subseteq E(G)$, and the goal is to find a heaviest induced matching $S$ of $G$ such that $S \subseteq Y$.
Then \mim is the particular case when $Y=E(G)$.
Of course, we could solely use the edge weights to emulate $Y$ (by giving negative weights to all the edges in $E(G) \setminus Y$).
We believe this formalism is slightly more convenient for the reader to quickly and explicitly identify where our algorithm is seeking mutually induced edges.

Since the case of \mim is more involved than were the treatment of \smis and \coloring, we again split the arguments into the design of a subexponential-time constant-approximation algorithm (\cref{lem:induced-matching-subexp}) followed by how this algorithm meets the requirements of~\cref{lem:induction-approx-fast} (\cref{lem:msim-iteration}). 

\begin{lemma}\label{lem:induced-matching-subexp}
  Assume every input $(G,Y)$ is given with a~\mbox{$d'$-sequence} of the $n$-vertex, edge-weighted by $w$, graph $G$.
  We set $d := c_{d'} \cdot 2^{4c_{d'}+4}$, and $s := 2^{4c_{d'}+4}$.
  \msim can be $O(d^2)$-approximated in time $2^{O_d(\sqrt n)}$ on these inputs.
\end{lemma}
\begin{proof}
  Again, by~\cref{lem:seq-to-partial-seq}, we start by computing in polynomial time a partition of $V(G)$, $\P=\{P_1, \ldots, P_{\lfloor \sqrt n \rfloor}\}$, of parts with size at most $s \sqrt n$ and such that $\mathcal R(G/\P)$ has maximum degree at most~$d$.

  We $(d+1)$-color $\mathcal R(G/\P)$, which defines a coarsening $\{C_1, \ldots, C_{d+1}\}$ of $\P$.
  We also distance-2-edge-color $\mathcal R(G/\P)$ with $z = 2(d-1)d+1$ colors, that is, properly (vertex-)color the square of its line graph.
  Observe that $z-1$ upperbounds the maximum degree of the square of the line graph of $\mathcal R(G/\P)$.
  This partitions the edges of $\mathcal R(G/\P)$ into $\{E_1, \ldots, E_z\}$.
  For each red edge $e=P_iP_j \in R(G/\P)$, we denote by $p(e)$ the set $P_i \cup P_j$.
  We also set $X_h = p(E_h) = \bigcup_{e \in E_h} p(e)$ for each $h \in [z]$.

  Let $M \subseteq Y$ be a fixed (unknown) heaviest induced matching of $G$ contained in $Y$.
  Let $M_v, M_r, M_b$ partition~$M$, where $M_v$ (as \textbf{v}ertex) consists of the edges of~$M$ with both endpoints in a same~$P_i$, $M_r$ (as \textbf{r}ed) corresponds to edges of~$M$ between some $P_i$ and $P_j$ with $P_iP_j \in R(G/\P)$, and $M_b$ (as \textbf{b}lack), the edges of $M$ between some $P_i$ and $P_j$ with $P_iP_j \in E(G/\P)$.  
  We compute three induced matchings $N_v, N_r, N_e \subseteq Y$ of $G$, capturing a~positive fraction of $M_v, M_r, M_e$, respectively. \cref{fig:matchings} gives the intuition of the procedures which determine each of these approximated solutions.
  
  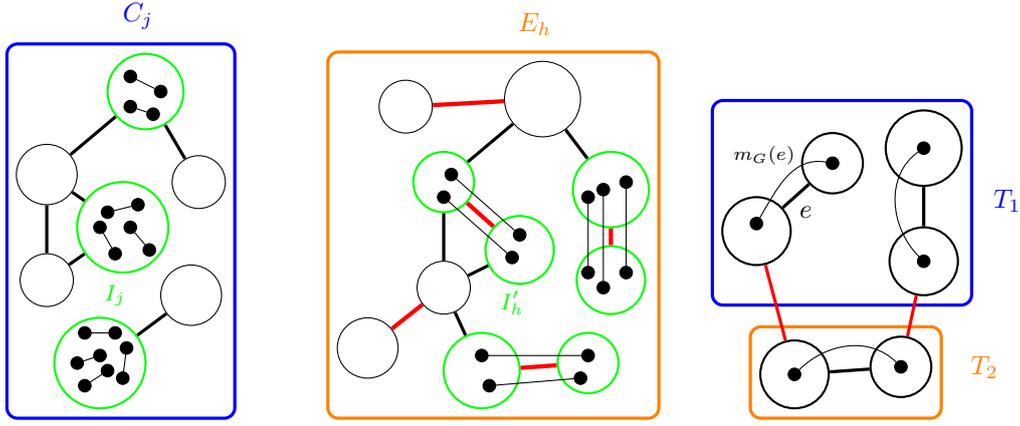
\begin{figure}[t]
   \centering
   \begin{subfigure}[b]{0.30\columnwidth}
   \centering
   \scalebox{1.0}{\begin{tikzpicture}
     \foreach \i/\j/\l/\s in {-3/0.9/q3/0.7, -5/1/q4/0.8, -5/-0.4/q5/0.7, -3.1/-0.6/q7/0.8}{
       \node[draw,circle,minimum size=\s cm] (\l) at (\i,\j) {} ;
     }
     \foreach \i/\j/\l/\s in {-4/0.3/q1/1.2, -3.7/2.1/q2/1,-4.3/-1.5/q6/1.2}{
       \node[draw,circle,color = green,thick,minimum size=\s cm] (\l) at (\i,\j) {} ;
     }
     \node[color = green] at (-4.1,-0.6) {\footnotesize{$I_j$}} ;

    \foreach \i/\j/\k/\l/\s/\t in {
       -3.5/2.1/-3.9/2.3/u1/v1, -3.6/1.8/-3.9/1.9/u2/v2,-3.8/0.6/-4.2/0.5/u3/v3,-4.3/0.3/-4.1/-0.05/u4/v4,-3.9/0.3/-3.65/0/u5/v5,-4.5/-1.1/-4.1/-1.1/u6/v6,-4.6/-1.5/-4.3/-1.4/u7/v7,-4.5/-1.8/-4.2/-1.6/u8/v8,-3.95/-1.3/-4.0/-1.7/u9/v9}{
       \node[draw,fill=black,circle,scale=0.5] (\s) at (\i,\j) {};
       \node[draw,fill=black,circle,scale = 0.5] (\t) at (\k,\l) {};
       \draw (\s) -- (\t) ;
     }
     
     \node[draw,very thick,blue,rounded corners,fit=(q1) (q2) (q3) (q4) (q5) (q6)] {} ;
     
     \node[color = blue] at (-3.8,3.1) {$C_j$} ;
     
     \foreach \i/\j in {
       q2/q3,q1/q4,q4/q5,q5/q1,q4/q2,q6/q7}{
       \draw[very thick] (\i) -- (\j) ;
     }

\end{tikzpicture}}
   \caption{Computing $R_j$ consists first of determining the heaviest induced matching in each part $P_i$ and then, for color $C_j$, to compute the maximum independent set $I_j$ (in green) weighted by the size of the matchings.}
   \end{subfigure}
   ~
   \begin{subfigure}[b]{0.35\columnwidth}
   \centering
   \scalebox{1.0}{\begin{tikzpicture}
     \foreach \i/\j/\l/\s in {-3.7/2.1/q2/1,  -5/-0.4/q5/0.7, -6/-1.2/q10/0.8, -5.5/2.0/q8/0.7}{
       \node[draw,circle,minimum size=\s cm] (\l) at (\i,\j) {} ;
     }
     \foreach \i/\j/\l/\s in {-5/1/q4/0.8,-4.0/0.1/q1/0.9,-4.5/-1.5/q6/1.0,-3.1/-1.4/q7/0.8,-2.8/0.9/q3/1.0, -2.8/-0.3/q9/0.9}{
       \node[draw,circle,color = green,thick,minimum size=\s cm] (\l) at (\i,\j) {} ;
     }
     \node[color = green] at (-4.1,-0.6) {\footnotesize{$I_h'$}} ;

    \foreach \i/\j/\k/\l/\s/\t in {
       -3.1/0.8/-3.1/-0.2/u1/v1, -2.9/0.9/-2.9/-0.4/u2/v2,-2.6/1.0/-2.6/-0.2/u3/v3,-4.5/-1.3/-3.1/-1.3/u4/v4, -4.4/-1.7/-3.2/-1.6/u5/v5,-4.9/1.1/-4/0.3/u6/v6,-5.0/0.8/-4.1/0.0/u7/v7}{
       \node[draw,fill=black,circle,scale=0.5] (\s) at (\i,\j) {};
       \node[draw,fill=black,circle,scale = 0.5] (\t) at (\k,\l) {};
       \draw (\s) -- (\t) ;
     }
     
     \node[draw,very thick,orange,rounded corners,fit=(q1) (q2) (q3) (q4) (q5) (q6) (q10)] {} ;
     
     \node[color = orange] at (-3.8,3.1) {$E_h$} ;
     
     \foreach \i/\j in {
       q2/q3,q4/q5,q5/q1,q4/q2,q6/q5}{
       \draw[very thick] (\i) -- (\j) ;
     }

     \foreach \i/\j in {
       q3/q9,q6/q7,q1/q4,q2/q8,q5/q10}{
       \draw[ultra thick, red] (\i) -- (\j) ;
     }
\end{tikzpicture}}
   \caption{Color $E_h$ reveals a set of red edges from trigraph $G/\P$. Set $R_h'$ corresponds to the heaviest matching among these edges which is mutually induced regarding the black edges. The weight of the red edges $e$ is $w(S_e')$.}
   \end{subfigure}
   ~
   \begin{subfigure}[b]{0.28\columnwidth}
   \centering
   \scalebox{1.0}{\begin{tikzpicture}
     \foreach \i/\j/\l/\s in {-4/1/q1/0.8,-5.0/0.1/q2/0.9,-4.5/-1.8/q3/0.9,-3.1/-1.7/q4/0.8,-2.8/1.2/q5/1.0, -2.8/-0.3/q6/0.9}{
       \node[draw,circle,thick,minimum size=\s cm] (\l) at (\i,\j) {} ;
     }
     \node[color = blue] at (-1.7,0.5) {$T_1$} ;
     \node[color = orange] at (-2.0,-1.7) {$T_2$} ;


\node[draw,fill=black,circle,scale=0.5] (u1) at (-4,1) {};
\node[draw,fill=black,circle,scale = 0.5] (v1) at (-5,0.2) {};
\draw (u1) to [out=170,in=60] (v1) ;

\node at (-4.35,0.35) {$e$} ;
\node at (-4.9,1.1) {\scriptsize{$m_G(e)$}} ;

\node[draw,fill=black,circle,scale=0.5] (u2) at (-4.5,-1.8) {};
\node[draw,fill=black,circle,scale = 0.5] (v2) at (-3.1,-1.7) {};
\draw (u2) to [out=45,in=135] (v2) ;

\node[draw,fill=black,circle,scale=0.5] (u3) at (-2.8,1.2) {};
\node[draw,fill=black,circle,scale = 0.5] (v3) at (-2.8,-0.3) {};
\draw (u3) to [out=215,in=135] (v3) ;
     
     \node[draw,very thick,blue,rounded corners,fit=(q1) (q2) (q5) (q6)] {} ;
     \node[draw,very thick,orange,rounded corners,fit=(q3) (q4)] {} ;
     
     
     \foreach \i/\j in {
       q1/q2,q3/q4,q5/q6}{
       \draw[very thick] (\i) -- (\j) ;
     }

     \foreach \i/\j in {
       q2/q3,q4/q6}{
       \draw[very thick, red] (\i) -- (\j) ;
     }
\end{tikzpicture}}
   \caption{An example of set $S$ of size 3 with two colors $T_1$ and $T_2$. The induced matching $R_i''$ for color $T_i$ is obtained by considering the maximum-weighted edge $m_G(e)$ between the two parts of $e$.}
   \end{subfigure}
   \caption{Illustration of how to determine the induced matching $N_v$, $N_r$, and $N_b$ (in that order, from left to right).}
   \label{fig:matchings}
  \end{figure}

  \medskip

  \textbf{Computing $N_v$.}
  For every integer $1 \leqslant i \leqslant \lceil \sqrt n \rceil$, we compute a heaviest induced matching in $G[P_i]$ contained in $Y$, say $S_i$, in time $2^{O_d(\sqrt n)}$.
  For each $j \in [d+1]$, let $H_j$ be the \emph{graph} $(G/\P)[C_j]$ with every vertex $P_i \in C_j$ weighted by $w(S_i)$.
  We compute a heaviest independent set $I_j$ in $H_j$, also in time $2^{O_d(\sqrt n)}$.

  Let $R_j$ be the induced matching $\{e \in S_i : P_i \in I_j\}$.
  It is indeed an induced matching in~$G$ contained in $Y$, since each $S_i$ is so, there is no red edge in $(G/\P)[C_j]$, and~$I_j$ is an independent set of~$H_j$.  
  The solution $N_v$ is then a heaviest among the~$R_j$s.

  \medskip
  
  \textbf{Computing $N_r$.}
  For each $e=P_iP_j \in R(G/\P)$, we compute a heaviest induced matching~$S'_e$ in $G[p(e)]=G[P_i \cup P_j]$ among those that are included in $Y$ and have only edges with one endpoint in $P_i$ and the other endpoint in $P_j$.
  This takes times at most $\frac{\sqrt n d}{2} \cdot 2^{O_d(\sqrt n)}=2^{O_d(\sqrt n)}$ by trying out all vertex subsets, since $|P_i \cup P_j| \leqslant 2s \sqrt n$.
  For each $h \in [z]$, let $H'_h$ be the \emph{graph} $(G/\P)[\{P_i~:~P_i~\text{is incident to an edge}~e \in E_h\}]$ and the red edges $e \in E_h$ are turned black and get weight $w(S'_e)$.
  We compute a heaviest induced matching $I'_h$ in $H'_h$ among those \emph{included in~$E_h$}, in time $2^{O_d(\sqrt n)}$.
  Note here that we changed the prescribed set of edges $Y$ to $E_h$.

  Let $R'_h$ be the induced matching $\{f \in S'_e : e \in I'_h\} \subseteq Y$ of $G$.
  Indeed, each $S'_e \subseteq Y$ is an induced matching, and there is no red edge between an endpoint of $e \in I'_h$ and an endpoint of $e' \neq e \in I'_h$ (since $E_h$ is a color class in a distance-2-edge-coloring of $\mathcal R(G/\P)$), nor a black edge (by virtue of $I'_h$ being an induced matching of $H'_h$).
  The solution $N_r$ is then a heaviest among the $R'_h$s.

  \medskip
  
  \textbf{Computing $N_b$.}
  Observe first that an induced matching of $G$ can only contain at most one edge between $P_i$ and $P_j$ when $P_iP_j \in E(G/\P)$.
  Thus in the graph $(V(G/\P),E(G/\P))$, we give weight $\max\{w(f) : f=uv \in Y, u \in P_i, v \in P_j\}$, with the convention that $\max \emptyset = -1$, to each edge $e=P_iP_j \in E(G/\P)$, call $G'$ the resulting edge-weighted graph, and denote by $m_G(e)$ an edge $f \in Y$ realizing this maximum.
  We compute a heaviest induced matching $S$ of $G'$ \emph{included in $E(G')$}, in time $2^{O_d(\sqrt n)}$.
  Let $H_S$ be the graph with vertex set $S$, and an edge between $e$ and $e'$ whenever there is a red edge in $G/\P$ between an endpoint of $e$ and an endpoint of $e'$.
  As $H_S$ has degree at most $2d$, it can be $2d+1$-colored; let $T_1, \ldots, T_{2d+1}$ the corresponding color classes.

  For each $i \in [2d+1]$, let $R''_i$ be the induced matching $\{m_G(e) : e \in T_i\} \subseteq Y$ of $G$.
  Indeed, $S$~is an induced matching in the black graph of $G/\P$, and the underlying vertices of $T_i$ do not induce any red edge in $G/\P$, by design.
  The solution $N_r$ is then a heaviest among the $R''_i$s.

  \medskip

  We finally output a heaviest set among $N_v, N_r, N_b$.
  The overall running time is $2^{O_d(\sqrt n)}$ as we make a polynomial number of calls to (exhaustive) subroutines on graphs with $O_d(\sqrt n)$ vertices, and color in linear time $O(n)$-vertex graphs of maximum degree $\Delta$ with $\Delta+1$ colors.
  We already argued that $N_v, N_r, N_b \subseteq Y$ are all induced matchings in $G$, thus so is our output.

  We shall just show that we meet the claimed approximation factor.
  First, one can observe $w(N_v) \geqslant \frac{w(M_v)}{d+1}$.
  Second, at least a $\frac{1}{z}$ fraction of the weight of $M_r$ intersects some fixed $E_i$ (with $i \in [z]$).
  Let $\mathcal J$ be the parts of $\P$ intersected by $M_r \cap X_i$.
  As there cannot be a black edge between two parts of $\mathcal J$ (otherwise $M_r$ is not an induced matching as defined), our algorithm indeed computes an induced matching of $G[X_i]$ included in $Y$ of weight at least $w(M_r \cap X_i)$.
  Hence $w(N_r) \geqslant \frac{w(M_r)}{z}$.

  Third, we already argued that an induced matching in $G'$ corresponds to an induced matching in the black graph of $G/\P$.
  Thus at least one of the $R''_i$ (with $i \in [2d+1]$) contains at least a $\frac{1}{2d+1}$ fraction of the weight of $M_b$.
  Therefore $w(N_b) \geqslant \frac{w(M_b)}{2d+1}$.

  Finally the output induced matching has at least weight \[\frac{w(M)}{3 \cdot \max(d+1,z,2d+1)}=\frac{w(M)}{3z}=\frac{w(M)}{3(2(d-1)d+1)}.\hfill \qedhere\]
\end{proof}

\begin{theorem}\label{lem:msim-iteration}
  \msim on an $n$-vertex graph $G$, edge-weighted by $w$, with prescribed set $Y \subseteq E(G)$, and given with a~$d'$-sequence, satisfies the assumptions of~\cref{lem:induction-approx-fast}.
  In particular, this problem admits
  \begin{compactitem}
  \item a $(d+1)^{2^q-1}$-approximation in time $2^{O_{d,q}(n^{2^{-q}})}$, for every integer $q \geqslant 0$,
  \item an $n^\varepsilon$-approximation in polynomial-time $O_{d,\varepsilon}(1) \log^{O_d(1)}n \cdot n^{O(1)}$, for any $\varepsilon > 0$, and
  \item a $\log n$-approximation in time $2^{O_d(n^{\frac{1}{\log \log n}})}$, 
  \end{compactitem}
  with $d := c_{2d'+2} \cdot 2^{4c_{2d'+2}+4}$.
\end{theorem}

\begin{proof}
  The exhaustive algorithm (trying out all vertex subsets and checking whether they induce a matching included in $Y$) solves \msim in time $2^{O(n)}$.
  Thus we show \msim satisfies the second item of~\cref{lem:induction-approx-fast}, as witnessed by~\cref{lem:induced-matching-subexp} where subcalls are dealt with recursively.
  We set $c_2 \geqslant 1$ as the required exponent to turn a $d'$-sequence into a~neatly divided matrix of $\mathcal M_{n,2d'+2}$, and compute the various needed colorings, the appropriate $\frac{1}{2} < c_1 < 1$, and $2 < c_3 < 3$, and $f(d)=2d \geqslant 1$ with $s := 2^{4c_{d'}+4}$.

  In computing $N_v$, the algorithm makes $\lfloor \sqrt n \rfloor$ recursive calls and $d + 1$ calls to \wis on induced subgraphs of $G$.
  All of these induced subgraphs are on less than $f(d) \sqrt n$ vertices.
  Computing $N_r$ makes at most $\frac{\sqrt n d}{2}$ recursive calls on induced subgraphs of $G$ with at most $f(d) \sqrt n$ vertices, followed by at most $2(d-1)d+1$ recursive calls on~\fullregus of induced subtrigraphs of $G/\P$ with at most $\sqrt n$ vertices (in fact, one can observe that the latter recursive calls happen to also be on induced subgraphs of $G$).
  Finally, computing $N_b$ makes one recursive call to a~\fullregu of $G/\P$ on $\lfloor \sqrt n \rfloor$ vertices.

  In summary, we make $O_d(\sqrt n)$ recursive calls or calls to another problem \swis (which already satisfies~\cref{lem:induction-approx-fast} with better constants) on induced subgraphs of $G$ or~\fullregus of (the whole) $G/\P$, each on $O_d(\sqrt n)$ vertices. 
  Hence, by \cref{lem:seq-to-partial-seq-induction}, the induction applies.

  We check that getting $r$-approximations on every subcall allows to output a global $3(2(d-1)d+1)r^2$-approximation.
  For that we argue that $N_v$  (resp., $N_r$, $N_b$) is a \emph{$(2(d-1)d+1)r^2$-approximation} of $M_v$ (resp., $M_r$, $M_b$).
  The fact that $N_v$ is a $(d+1)r^2$-approximation (hence a $(2(d-1)d+1)r^2$-approximation, since we assume that $d \geqslant 1$) of $M_v$ directly follows~\cref{lem:mis-iteration}.

  We now show that $N_r$ is a $(2(d-1)d+1)r^2$-approximation of $M_r$.
  Let $h \in [z] = [2(d-1)d+1]$ be an index maximizing $w(M_r \cap E(G[X_h]))$.
  Thus $w(M_r \cap E(G[X_h])) \geqslant \frac{w(M_r)}{2(d-1)d+1}$.
  Let $F_h \subseteq E_h$ be the edges $e=P_iP_j$ of $\mathcal R(G/\P)$ that are inhabited by $M_r$ (i.e., $M_r$ contains at least one edge between $P_i$ and $P_j$).
  Note that our algorithm makes an $r$-approximation of the optimum such solutions on $p(e)$ (selecting only edges between $P_i$ and $P_j$).
  Thus the $r$-approximation on $H'_h$ yields the desired $(2(d-1)d+1)r^2$-approximation $N_r$.

  Finally, one can easily see that $N_b$ is a $(2d+1)r$-approximation of $M_b$ (note, here, the absence of a 2 in the exponent of $r$).
\end{proof}

\section{Technical generalizations}\label{sec:technical}

\subsection{\mihph}

In this section we present a far-reaching generalization of the approximation algorithms for \mis and \mim.
For any fixed graph $H$, let \mihp{H} be the problem where one seeks a largest collection of mutually induced copies of $H$ in the input graph $G$, that is, a largest set $S$ such that $G[S]$ is a disjoint union of (copies of) graphs $H$.
We get similar approximation guarantees for \mihp{H}, for any connected graph~$H$.
Observe that \mis and \mim are the special cases when $H$ is a single vertex and a single edge, respectively.



We in fact approximate a technical generalization that we call \aihp{H}.
The input is a tuple $(G,w,z,\gamma,\gamma')$ where $G$ is a graph, $w: V(G)^{|V(H)|} \to \mathbb Q$ is a weight function over the tuples \emph{without repetition} of $V(G)$ of size $|V(H)|$ (that we will use to keep track of the number of mutually induced copies \emph{within} a given tuple of vertices of $G$), $z$ is an integer between 1 and $|V(H)|$, $\gamma: V(G) \to [z]$ is a labeled partition of $V(G)$ into $z$ classes, and $\gamma': V(H) \to [z]$ is a labeled partition of $V(H)$ into $z$ classes.
Note that the \mihp{H} is obtained when $w(Z)=[G[Z]~\text{is isomorphic to}~H]$ (where $[.]$ is the Iverson bracket, i.e., taking value 1 if the property it surrounds is true, and 0 otherwise) and $z=1$ (which forces the value of $\gamma$ and $\gamma'$).
The goal is to find a subset $S$ such that
\begin{compactitem}
\item $G[S]$ is a disjoint union of copies of $H$,
\item there is an isomorphism between each copy $C$ of $H$ (in $S$) and $H$ which preserves $\gamma, \gamma'$, i.e., every vertex $v$ of $C$ is mapped to a vertex $v' \in V(H)$ with $\gamma(v)=\gamma'(v')$, and 
\item $\sum\limits_{C~\text{copy of}~H~\text{in}~S} w(V(C))$ is maximized.
\end{compactitem}
  
We will need the notion of \emph{compatible trigraphs} of a (labeled) graph.
Given a graph $H$, we call \emph{compatible trigraph of $H$} any trigraph on at most $|V(H)|$ vertices obtained by turning some (possibly none) black edges or non-edges of trigraph $H/\mathcal Q$ (for any fixed choice of a partition $\mathcal Q$ of $V(H)$) into red edges.
In other words, a compatible trigraph $H'$ of $H$ is such that there is a \regu $H''$ of $H'$ that is also a quotient trigraph of~$H$.
Note that the number of compatible trigraphs of an $h$-vertex graph $H$ is upperbounded by $B_h \cdot 2^{{h \choose 2}}=2^{O(h^2)}$, where $B_h$ is the $h$-th Bell number, which counts the number of partitions of a set of size $h$.

Given a~graph $G$ vertex-partitioned by~$\P$ and a trigraph $H$, a~subset $S \subseteq V(G)$ is said \emph{cut by $\P$ along $H$} if $G[S]/\P$ is isomorphic to $H$.
By extension, the copy of $G[S]$ in $G$ (induced by $S$) is also said cut by $\P$ along $H$.

\begin{lemma}\label{lem:induced-matching-subexp}
  For any connected graph $H$, \aihp{H}, when every input $(G,w,z,\gamma,\gamma')$ is given with a~\mbox{$d'$-sequence} of the $n$-vertex graph $G$,
  satisfies the assumptions of~\cref{lem:induction-approx-fast}.
  In particular, this problem admits
  \begin{compactitem}
  \item a $d^{O_h(2^q)}$-approximation in time $2^{O_{d,h,q}(n^{2^{-q}})}$, for every integer $q \geqslant 0$,
  \item an $n^\varepsilon$-approximation in polynomial-time $O_{\varepsilon}(1) \cdot n^{O_{d,h}(1)}$, for any $\varepsilon > 0$,
  \end{compactitem}
   with $h = |V(H)|$, and $d := c_{2d'+2} \cdot 2^{4c_{2d'+2}+4}$.
\end{lemma}
\begin{proof}
  As the first item of~\cref{lem:induction-approx-fast} is satisfied, we describe an algorithm that fulfills the requirement of its second item.
  We proceed by induction on the number of vertices of $H$.
  Thus we can assume that \aihp{J}, with $J$ a~connected graph on less vertices than $H$, satisfies~\cref{lem:induction-approx-fast}.
  We already did the base case of the induction, which was \wis.

\medskip
  
  \textbf{Algorithm.}
  Again, by~\cref{lem:seq-to-partial-seq-induction}, we start by computing in polynomial time a partition of $V(G)$, $\P=\{P_1, \ldots, P_{\lfloor \sqrt n \rfloor}\}$, of parts with size at most $d \sqrt n$ and such that $\mathcal R(G/\P)$ has maximum degree at most~$d$.
  Let $S$ be a fixed (unknown) heaviest (with respect to $w$) mutually induced $H$-Packing of $G$ preserving $\gamma, \gamma'$.

  For every compatible trigraph $H'$ of~$H$, we look for mutually induced copies of~$H$ in~$G$ cut by~$\P$ along $H'$, and preserving $\gamma, \gamma'$.
  As the number of compatible trigraphs of $H$ is $2^{O(h^2)}$, a~$1/2^{O(h^2)}$ fraction of the~weight of~$S$ is made of mutually induced copies of $H$ which are cut by $\P$ along a fixed compatible trigraph $H'$.
  We now focus on this particular ``run.''

  We distinguish two cases:
  \begin{compactitem}
  \item (A) $H'$ has at least one black edge, or
  \item (B) $H'$ has no black edge.
  \end{compactitem}
  As $H$ is connected, the total graph of $H'$ is also connected.
  Indeed, switching some edges or non-edge to red edges in the quotient trigraph of $H$ cannot disconnect the total graph, which can only gain edges.
  Thus in case (A), every red component of $H'$ has at least one incident black edge, and in case (B), $H'$ has a single red component (and no black edge).

  In general, we want to individually pack red components of $H'$ (first type of recursive calls in smaller induced subgraphs of $G$), then combine those red components by connecting them with the right pattern of black edges (second type of recursive calls in the total graph of $G/\P$).  
  Handling both cases (A) and (B) in an unified way runs into the~technical issue that the weight function may destroy our combined solutions in an uncontrollable manner.
  The case distinction works as a win-win argument.
  In case~(A), due to the presence of a black edge in $H'$, we can pack at most one mutually induced copy of $H$ within any fixed subtrigraph of $G/\P$ matching $H'$.
  We thus exempt ourselves from the first type of recursive calls.
  In case~(B), we do need the two types of recursive calls (as in \swis), but the first type is done on the whole $H$.
  Thus the current weight function (on $h$-tuples) is informative enough.
  
  \textbf{Case (A).}
  The essential element here is to build a new weight function $w'$ on the $h'$-tuples of the total graph $\mathcal T(G/\P)$, with $h' := |V(H')|$.
  For every injective map $\iota: V(H') \to \P$ inducing a trigraph isomorphism and preserving $\gamma, \gamma'$, for every ordering of $\iota(V(H'))$ into an $h'$-tuple $(P_1,\ldots,P_{h'})$, we set $$w'(P_1,\ldots,P_{h'}) := \max\{w(v_1^1,v_1^2,\ldots,v_1^{a_1}, \ldots, v_{h'}^1,v_{h'}^2,\ldots,v_{h'}^{a_{h'}})~:~v_1^1,v_1^2,\ldots,v_1^{a_1} \in P_1, \ldots$$
  $$v_{h'}^1,v_{h'}^2,\ldots,v_{h'}^{a_{h'}} \in P_{h'},~\text{and}~G[\{v_1^1,v_1^2,\ldots,v_1^{a_1}, \ldots, v_{h'}^1,v_{h'}^2,\ldots,v_{h'}^{a_{h'}}\}]~\text{is isomorphic to}~H\}.$$ 
  Indeed as we previously observed, in case (A), at most one mutually induced copy of $H$ respecting the cut along $H'$ can be packed in the subgraph of $G$ induced by the vertices of $\iota(V(H'))$.
  (In the definition of $w'$, we can further impose that $a_i$ matches the number of vertices of $H$ in the corresponding part of $H'$ but this is not necessary.)

  All the $h'$-tuples not getting an image by $w'$ in the previous loop (realized in time $n^{O(h)}$) are assigned the value 0.
  We then make a recursive call to \aihp{\mathcal T(H')} on input $(\mathcal T(G/\P),w',1,\gamma_0,\gamma'_0)$ where we recall that $\mathcal T(.)$ is the total graph, and $\gamma_0,\gamma'_0$ are the constant 1 functions.

    \textbf{Case (B).}
  For every injective map $\iota: V(H') \to \P$ inducing a trigraph isomorphism and preserving $\gamma, \gamma'$, we make a~recursive call to \aihp{H} with input $(G_\iota = G[\bigcup_{P \in \iota(V(H'))} P],w,h,\gamma_\iota,\gamma'_\iota)$ where two vertices get the same label by $\gamma_\iota$ if and only if they have the same label by $\gamma$ and lie in the same $P \in \iota(V(H'))$, and $\gamma'_\iota$ gives to a vertex $v' \in X \in V(H')$ of $H$ the same label given to the vertices $v \in \iota(X)$ such that $\gamma'(v')=\gamma(v)$.
  Informally $\gamma_\iota, \gamma'_\iota$ forces the recursive call to commit to the map $\iota$ and the former functions $\gamma, \gamma'$.

  Each such recursive call yields a~mutually induced packing of $H$.
  Since the red graph of $G/\P$ has degree at most~$d$, we can color the (ordered) tuples of $\P$ of length up to $h$ and inducing a connected subgraph of $\mathcal R(G/\P)$ with at most $p(h,d) = h d^{2h} \cdot d^{2h} \cdot h! + 1$ colors such that every color class consists of disjoint tuples pairwise not linked by a red edge in $G/\P$.
   Indeed the claimed number of colors minus 1 upperbounds, in $\mathcal R(G/\P)$, the number of connected tuples of length up to $h$ that can touch (i.e., intersect or be adjacent to) a fixed connected tuple of length up to $h$.
   One color class contains a fraction $1/p(h,d)$ of the weight of the optimal solution $S$ (subject to the same constraints).
   Running through all color classes $j$ (and focusing on one containing a largest fraction of the optimum), we define a weight function $w'$ on the $h'$-tuples of $\mathcal T(G\P)$, with $h' = |V(H')|$, by giving to a tuple the weight returned by the corresponding recursive call whenever it is part of color class $j$, and weight 0 otherwise.
    We then make a recursive call to \aihp{\mathcal T(H')} on input $(\mathcal T(G/\P),w',1,\gamma_0,\gamma'_0)$ where we recall that $\mathcal T(.)$ is the total graph, and $\gamma_0,\gamma'_0$ are the constant 1 functions.

   We output a heaviest solution among all runs.
   We now check that the algorithm is as prescribed by~\cref{lem:induction-approx-fast}.

   \medskip
   
   \textbf{Number of recursive calls.}
   We make at most $2^{O(h^2)} \cdot h \cdot |V(G/\P)|^h=n^{O_h(1)}$ recursive calls \emph{to \aihp{H}}, and at most $p(h,d)+1=O_{d,h}(1)$ recursive calls \emph{to \aihp{\mathcal T(H')}}.
   Hence there is a constant $c_1$ (function of $d$ and $h$) such that the number of calls is bounded by $n^{c_1}$.

   \medskip

   \textbf{Nature and size of the inputs of the recursive calls.}
   Both $H$ and $\mathcal T(H')$ have strictly less vertices than $H$ or are equal to $H$.
   Thus the induction on $h$ applies.
   Besides, $G[\bigcup_{P \in \iota(V'(H))} P]$ is an induced subgraph of $G$ of size at most $h \cdot d \sqrt n = O_{d,h}(1) \cdot \sqrt n$, and $\mathcal T(G/\P)$ is a~\fullregu of $G/\P$ of size at most $\lfloor \sqrt n \rfloor$.

   \medskip

   \textbf{Running time.}
   Outside of the recursive calls, one can observe that our algorithm takes times $O_{d,h}(1) \cdot n^{O_h(1)}$.
   Hence there is a constant $c_2$ (function of $d$ and $h$) such that the running time of that part is bounded by $n^{c_2}$.

   \medskip

   \textbf{Correctness and approximation guarantee.}
   As all the recursive calls are on induced subgraphs of $G$ or of the total graph $\mathcal T(G/\P)$, we return a mutually induced collection of graphs of the size of $H$.
   All these graphs are indeed induced copies of $H$ since the weight function prevents the false positives of copies of $H$ in the total graph $\mathcal T(G/\P)$ but not in $G$ (these tuples are given weight 0).
   Finally it can be checked that the returned solution has weight a fraction $(2^{O(h^2)} \cdot \max(r,p(h,d)r^2))^{-1}$ of the optimum, which can also be seen as a~$d^{c_3}r^2$-approximation for some constant $c_3$ depending on $d$ and $h$.
\end{proof}

\subsection{Independent induced packing of stars and forests}

The techniques employed to design approximations algorithms for \msim can be extended in order to tackle more general problems.
In particular, we show in this section a generalization of \cref{lem:msim-iteration} for \misf and \mif.
These two problems stand as the version of \mihp{\mathcal{H}} where $\mathcal{H}$ is respectively either the infinite family of stars or trees.

On the one hand, \misf asks, given a graph $G$ and a subset $Y \subseteq E(G)$, for a collection of induced stars on $G$, made up of edges of $Y$ only, maximizing the number of edges (or leaves). 

\defoptproblem{\misf}{Graph $G$, subset $Y \subseteq E(G)$}{
Collection $(A_i)_{i \in \left[ k \right]}$ of induced stars on $G$, made up of edges in $Y$ only, such that there is no edge between $A_i$ and $A_j$, for any $i \neq j \in [k]$, which maximizes the number of edges.
}

On the other hand, given the same input, \mif asks for an induced forest $F$ on $G$ with the largest set of edges.

We would like to emphasize the fact that the objective function of both problems counts the number of \textit{edges} in the solution, instead of \textit{vertices}, as it is often the case in the literature when looking for a collection of stars or trees in a graph. The reason for this is because an approximated solution for these vertex versions can be obtained from an approximated solution of \wis (since any independent set is a star forest, and any forest is a bipartite graph).

 Observe moreover that a solution of \mif can be 3-approximated with a solution of \misf. Indeed, the edge set of any tree can be partitioned into three distance-2-edge colors, which consist of a collection of stars. Therefore, the induced forest $F$ can be partitioned into three collections of induced stars. In the remainder, we design approximation algorithms for \misf, and directly deduce results for \mif.



In the remainder, we propose approximation algorithms for \misf. 
We provide in particular a $n^{\varepsilon}$-approximation algorithm for \misf, running in polynomial time.

We need to find the suitable generalization of \misf, as it was done for \textsc{Coloring} in~\cref{subsec:coloring}. We call this problem \mlisf. Now, a weight function on vertices is added to the input, and we seek a collection of mutually induced stars with maximum weight, the weight of a star being the sum of the weights of its leaves (that is, the weight of the root is omitted).

\defoptproblem{\mlisf}{Graph $G$, weights $w_V: V \rightarrow \mathbb{N}$, subset $Y \subseteq E(G)$}{
Collection $(A_i)_{i \in \left[ k \right]}$ of induced stars on $G$ with root $r_i$, $A_i = \{r_i, s_i^{1},\ldots,s_i^{L_i} \}$, made up of edges in $Y$ only, with no edge between $A_i$ and $A_j$, for any $i \neq j \in [k]$, maximizing
$$\sum_{i=1}^k w_V(A_i) = \sum_{i=1}^k \sum_{\ell=1}^{L_i} w(s_i^{\ell})$$
}


%

We prove that \mlisf follows the framework proposed in~\cref{lem:induction-approx-fast}. We begin with the design of a subexponential-time algorithm approximating a solution of \mlisf with a ratio function of twin-width.

\begin{lemma}\label{lem:star-forest-subexp}
  Assume every input of \mlisf is given with a~\mbox{$d'$-sequence} of the $n$-vertex $G$, and $d := c_{2d'+2} \cdot 2^{4c_{2d'+2}+4}$.
  \mlisf can be $O(d^2)$-approximated in time $2^{O_d(\sqrt n)}$ on these inputs.
\end{lemma}
\begin{proof}
We compute in polynomial time a partition of $V(G)$, $\P=\{P_1, \ldots, P_{\lfloor \sqrt n \rfloor}\}$, of parts with size at most $d \sqrt n$ and such that $\mathcal R(G/\P)$ has maximum degree at most~$d$, by \cref{lem:seq-to-partial-seq}.

  As in~\cref{lem:mis-subexp-approx}, we $(d+1)$-color $\mathcal R(G/\P)$, which defines a coarsening $\{C_1, \ldots, C_{d+1}\}$ of~$\P$.
  Moreover, we distance-2-edge-color $\mathcal R(G/\P)$ with $z = 2(d-1)d+1$ colors.
  This partitions the edges of $\mathcal R(G/\P)$ into $\{E_1, \ldots, E_z\}$.
  For each red edge $e=P_iP_j \in R(G/\P)$, we denote by $p(e)$ the set $P_i \cup P_j$.

  Let $A = \bigcup_{i=1}^k A_i$ be the union of all stars present in an optimum solution of \mlisf in $G$. We have $A \subseteq Y$.
  Let $A_v, A_r, A_b$ partition~$A$, where $A_v$ contains the edges of~$A$ with both endpoints in a same~$P_i$, $A_r$ corresponds to edges of~$A$ between some $P_i$ and $P_j$ with $P_iP_j \in R(G/\P)$, and $A_b$, the edges of $A$ between some $P_i$ and $P_j$ with $P_iP_j \in E(G/\P)$. The set of edges $A_v$ (resp. $A_r$, $A_b$) still form a collection of mutually induced stars. At least one over the three solutions produced by the partition $A_v, A_r, A_b$ gives us a 3-approximation for this problem. 
  Our algorithm consists of computing three solutions for \mlisf of $G$, capturing a~positive fraction of $A_v, A_r, A_b$, respectively.
  
  \textbf{Computing a} $d+1$-\textbf{approx for} $A_v$. \textit{Construction}. For every integer $1 \leqslant i \leqslant \lceil \sqrt n \rceil$, we compute an optimum solution for \mlisf in $G[P_i]$ contained in $Y$, say $S_i$, in time $2^{O_d(\sqrt n)}$. This can be achieved with guesses of the vertices in $P_i$, as $\vert P_i \vert \le d\sqrt{n}$.
 

Then, we focus on each color $C_j$ of $\mathcal R(G/\P)$, for $j \in [d+1]$. There is no red edge in $H_j = (G/\P)\left[C_j\right]$. We compute a heaviest independent set $I_j$ in $H_j$ where the parts $P_i$ are weighted by the edge weight of $S_i$. Let $R_j$ be the union of all optimum solutions for \mlisf on all $P_i$ belonging to $I_j$. The solution returned is the maximum over all $R_j$s.

\textit{Approximation ratio}. Let $A_v^j$ be the subset of $A_v$ made up of edges belonging to parts of $C_j$. There is no red edge between two parts of $C_j$, therefore their neighborhood consists of either full adjacency or full non-adjacency. As a consequence, a maximum-weighted collection of stars in $C_j$ with edges inside parts intersects parts which are pairwise non-adjacent in $(G/\P)[C_j]$, otherwise the stars are not mutually induced. Consequently, this justifies that the set $R_j$ returned for each $C_j$ is a maximum-weighted collection of stars in $C_j$ made up of edges inside parts. In summary, the weight of each collection $R_j$ is greater than the weight of $A_v^j$. As $j \in [d+1]$, a heaviest collection among all $R_j$s is a $d+1$-approximation of $A_v$.
  
\textbf{Computing a} $O(d^2)$-\textbf{approx for} $A_r$. \textit{Construction}. For each $e=P_iP_j \in R(G/\P)$, we compute an optimal solution for \mlisf in $G[p(e)]=G[P_i \cup P_j]$ among those that are included in $Y$ and have only edges with one endpoint in $P_i$ and the other endpoint in $P_j$. Said differently, we determine a maximum-weighted collection of induced stars in $G[p(e)]$ over $Y$ with a root on one side (for example, $P_i$) and all leaves on the other side ($P_j$). This costs at most $2^{O_d(\sqrt n)}$ by trying out all vertex subsets, since $\vert P_i \cup P_i \vert \le 2d\sqrt n$. The set of vertices of the solution returned on $G[p(e)]$ is denoted by $B_e \subseteq p(e)$.

  For each $h \in [z]$, let $H'_h$ be the trigraph $(G/\P)[\{P_i~:~P_i~\text{is incident to an edge}~e \in E_h\}]$. The red edges of $H'_h$ form an induced matching on the red graph of $H'_h$ as they are at distance 2 in $G/\P$.
  We associate with any edge $e \in E_h$ the edge weight of $B_e$. Then, we turn the red edges of $H_h'$ in black: let $H''_h$ be the graph obtained. We solve \msim on $H''_h$ by restricting it to edges of $E_h$ (which plays the role of $Y$): this is achieved in $2^{O(\sqrt{n})}$ as $\vert V(H''_h) \vert \le \sqrt{n}$. Let $I''_h$ be a maximum-weighted induced matching obtained. For each $h \in [z]$, we obtain the union $R_h$ of all $B_e$, $e \in I''_h$: $R_h = \bigcup_{e \in I''_h} B_e$. We return an $R_h$ which maximizes the total edge weight, among all $h \in [z]$.
  
  \textit{Approximation ratio}. Let $A_r^h$ be the subset of $A_r$ made up of edges being part of red edges $E_h$ in $G/\P$, for $h \in [z]$. As the edges of $E_h$ form an induced matching in $\mathcal{R}(G/\P)$, the union of solutions of \mlisf over graphs $G[p(e)]$ with $e \in E_h$ can only be connected through black edges of $G/\P$. Furthermore, two collections of stars over $G[p(e)]$ and $G[p(f)]$ are necessarily not mutually induced if there is a black edge between an endpoint of $e$ and an endpoint of $f$. Consequently, $R_h$ gives a maximum-weighted collection of mutually induced stars over $E_h$ and its weight is at least the weight of $A_r^h$. The maximum-weighted collection over all $R_h$ gives a $z$-approximation, as $h \in [z]$.
  
  \textbf{Computing a} $2d+1$-\textbf{approx for} $A_b$. \textit{Construction}. For each part $P_i$, we solve \wis on $G[P_i]$ with weight function $w_V$. Let $I(P_i)$ be the independent set returned and $w(P_i)$ its weight. We focus now on graph $G' = (V(G/\P),E(G/\P))$, made up of the black edges of $G/\P$, and solve \mlisf on it with weights $w(P_i)$. As $\vert V(G') \vert \le \sqrt{n}$, this is achieved in $2^{O(\sqrt{n})}$. 
  
  Let $(B_h)_{h \in [k]}$ be the collection of stars returned, $B_h = \{R_h,S_h^{1},\ldots,S_h^{L_h}\}$ and $B \in E(G')$ be the set of edges belonging to this collection. 
  Based on the bounded maximum red degree of $G/\P$, we determine a $O(d)$-partition of the edges of $B$, in order to produce collections of mutually induced stars. Let $H^*$ be the graph where each edge $e$ in the collection $(B_h)_{h \in [k]}$ is represented with a vertex and two of them $e,f$ are adjacent if and only if there is a red edge in $G/\P$ connecting an endpoint of $e$ with an endpoint of $f$. This graph has degree at most $2d$, so it can be $2d+1$-colored: let $T_1,\ldots,T_{2d+1}$ be the corresponding color classes. Any set of edges $T_j$ gives us a collection of mutually induced stars on trigraph $G/\P$, in the sense that there is neither a black nor a red edge between two stars.
  
  We fix some color class: say $T_1$ w.l.o.g. Let $(B_h^*)$ be the collection of stars produced by $T_1$, where $B_h^* = \{R_h^*,S_h^{1,*},\ldots, S_h^{L_h^*,*}\}$. For the root $R_h^* = P_i$ of each star $B_h^*$, we select an arbitrary vertex $r_h \in P_i$. Let $(B_h^**)_{h \in [k]}$ be the following collection of stars (which are mutually induced) on $G$: $B_h^{**} = \{r_h\} \cup \bigcup_{\ell = 1}^{L_h^*} I(S_h^{\ell,*})$. In brief, the collection $(B_h^{**})_{h \in [k]}$ is made up of an arbitrary vertex of each root of stars $B_h^*$ and a maximum-weighted independent set of each leaf of $B_h^*$. Remember that we computed this collection of stars for $T_1$: we return a maximum-weighted collection $(B_h^{**})_{h \in [k]}$ among all the ones determined for $T_j$, $j \in [2d+1]$.

  \textit{Approximation ratio}. Any collection $B_b$ with stars belonging only to black edges of $G/\P$ reveals a collection of stars on the quotient graph. Concretely, two black edges of $G/\P$ containing each a branch of $B_b$ must be either non-adjacent or form an induced 3-vertex path on $G' = (V(G/\P),E(G/\P))$. Conversely, considering a collection $B^*$ of mutually induced stars of $G'$ and, for each $e \in B^*$, a collection $B_e^*$ of  mutually induced stars on $G[p(e)]$ produces a global collection of stars of $G$: then, we can partition its edges into $2d+1$ parts (as with $T_1,\ldots,T_{2d+1}$) such that each part contains mutually induced stars. As the collection $B$ computed above provides us with a heaviest collection of $G'$, a maximum-weighted $B_h^{**}$ over all $T_j$ is a $2d+1$-approximation for $B$, whose weight is at least the weight of $A_b$.
  
  \textbf{Conclusion of the proof}. We finally output a heaviest collection of mutually induced stars among the three approximating respectively $A_v$, $A_r$, and $A_b$. The overall running time is in $2^{O_d(\sqrt{n})}$. An upper bound for the approximation ratio of this algorithm is $3z = O(d^2)$.
\end{proof}

As for the other problems treated in this article, we apply to \mlisf the time-approximation trade-off proposed in~\cref{lem:induction-approx-fast}.

\begin{theorem}\label{lem:forest-iteration}
  \mlisf on an $n$-vertex graph $G$, weight function $w_V$, with prescribed set $Y \subseteq E(G)$, and given with a~$d'$-sequence, satisfies the assumptions of~\cref{lem:induction-approx-fast}.
  In particular, this problem admits
  \begin{compactitem}
  \item a $(d+1)^{2^q-1}$-approximation in time $2^{O_{d,q}(n^{2^{-q}})}$, for every integer $q \geqslant 0$,
  \item an $n^\varepsilon$-approximation in polynomial-time $O_{d,\varepsilon}(1) \log^{O_d(1)}n \cdot n^{O(1)}$, for any $\varepsilon > 0$, and
  \item a $\log n$-approximation in time $2^{O_d(n^{\frac{1}{\log \log n}})}$, 
  \end{compactitem}
  with $d := c_{2d'+2} \cdot 2^{4c_{2d'+2}+4}$.
\end{theorem}
\begin{proof}
The exhaustive algorithm (trying out all vertex subsets and checking whether they induce a collection of mutually induced stars in $Y$) solves \mlisf in time $2^{O(n)}$.
  Thus we show \mlisf satisfies the second item of~\cref{lem:induction-approx-fast}.
  We set $c_2 \geqslant 1$ as the required exponent to turn a $d'$-sequence into a~neatly divided matrix of $\mathcal M_{n,2d'+2}$ conform to~$G$, and compute the various needed colorings, the appropriate $\frac{1}{2} < c_1 < 1$, and $2 < c_3 < 3$, and $f(d)=2d \geqslant 1$.

  \textbf{Approximating $A_v$}. The algorithm makes $\lfloor \sqrt n \rfloor$ recursive calls to solve \mlisf on parts $P_i$. Furthermore, $d + 1$ calls to \swis are needed on induced subgraphs of $G/\P$.
  All of these induced subgraphs are on at most $d \sqrt n$ vertices.
  
  \textbf{Approximating $A_r$}. The algorithm makes at most $\frac{\sqrt n d}{2}$ recursive calls (one call per red edge of $G/\P$) on induced subgraphs of $G$ with at most $2d \sqrt n$ vertices, followed by at most $2(d-1)d+1$ calls of \msim on~\fullregus of induced subtrigraphs of $G/\P$ with at most $\sqrt n$ vertices.
  
  \textbf{Approximating $A_b$}. The algorithm makes $\lfloor \sqrt n \rfloor$ calls to solve \swis on parts $P_i$ and one recursive call on a~\fullregu of $G/\P$ on $\lfloor \sqrt n \rfloor$ vertices.

  In summary, we make $O_d(\sqrt n)$ recursive calls or calls to problems \swis and \msim (which already satisfy~\cref{lem:induction-approx-fast} with better constants) on induced subgraphs of $G$ or~\fullregus of (the whole) $G/\P$, each on $O_d(\sqrt n)$ vertices. 
  Hence, by \cref{lem:seq-to-partial-seq-induction}, the induction applies.

  Getting $r$-approximations on every subcall allows us to output a global $3(2(d-1)d+1)r^2$-approximation for \mlisf:
  \begin{compactitem}
  \item collection $A_v$ is approximated with ratio $(d+1)r^2$
  \item collection $A_r$ is approximated with ratio $(2(d-1)d+1)r^2$
  \item collection $A_b$ is approximated with ratio $(2d+1)r^2$.
  \end{compactitem}
The extra factor 3 comes from the fact that we output the heaviest of these three solutions.
\end{proof}

\misf is a particular case of \mlisf with $w_V(u) = 1$ for every vertex $u \in V(G)$. 
 Furthermore, a solution of \misf is a 3-approximation of a solution of \mif. These observations together with~\cref{lem:forest-iteration} allow us to state the following result.
 
\begin{corollary}
\misf and \mif on an $n$-vertex graph $G$, with prescribed set $Y \subseteq E(G)$, and given with a~$d'$-sequence, admit
\begin{compactitem}
  \item an $n^\varepsilon$-approximation in polynomial-time $O_{d,\varepsilon}(1) \log^{O_d(1)}n \cdot n^{O(1)}$, for any $\varepsilon > 0$, and
  \item a $\log n$-approximation in time $2^{O_d(n^{\frac{1}{\log \log n}})}$, 
\end{compactitem}
with $d := c_{2d'+2} \cdot 2^{4c_{2d'+2}+4}$.
\label{co:mif-approx}
\end{corollary}

\section{Limits}\label{sec:limits}

We now discuss the limits of our framework.
We give some examples of problems that are unlikely to have an $n^\varepsilon$-approximation algorithm on graphs of bounded twin-width.
The first such problem is \mids, where one seeks a~minimum-cardinality set which is both an independent set and a~dominating set.
In general $n$-vertex graphs, this problem cannot be $n^{1-\varepsilon}$-approximated in polynomial time unless P$=$NP~\cite{Halldorsson93}, and cannot be $r$-approximated in time $2^{o(n/r)}$ for any $r=r(n)$, unless the ETH fails~\cite{BonnetLP18}.

We show that \mids has the same polytime inapproximability in graphs of bounded twin-width.

\begin{theorem}\label{thm:inapprox-mids}
For every $\varepsilon > 0$, \mids cannot be $n^{1-\varepsilon}$-approximated in polynomial time on $n$-vertex graphs of twin-width at most~9 given with a 9-sequence, unless P$=$NP.   
\end{theorem}
\begin{proof}
  We perform the classic reduction of Halldórsson from~\textsc{SAT}~\cite{Halldorsson93}, but from \textsc{Planar 3-SAT} where each literal has at most two occurrences, which remains NP-complete~\cite{Kozma12}.
  More precisely we add a triangle $d_i, t_i, f_i$ for each variable $x_i$ (with $i \in [N]$), and an independent set $I_j$ of size $r$ for each 3-clause $C_j$ (with $j \in [M]$).
  We link $t_i$ to all the vertices of $I_j$ whenever $x_i$ appears positively in $C_j$, and we link $f_i$ to all the vertices of $I_j$ whenever $x_i$ appears negatively in $C_j$.
  This defines a graph $G$ with $n = 3N+rM$ vertices.
  
  It can be observed that if the \textsc{Planar 3-SAT} instance is satisfiable, then there is an independent dominating set of size $N$, whereas if the formula is unsatisfiable then any independent dominating set has size at least $r$.  
  Setting $r := N^{\frac{2-\varepsilon}{\varepsilon}}$, the gap between positive and negative instances is $\Theta_\varepsilon(1) n^{1-\varepsilon}$, while preserving the fact that the reduction is polynomial.

  Let us now argue that $G$ has twin-width at most~9, and that a 9-sequence of it can be computed in polynomial time.
  We can first contract each $I_j$ into a single vertex without creating a red edge.
  Next we can contract every triangle $d_i, t_i, f_i$ into a single vertex of red degree at most~4.
  At this point, the current trigraph is a planar graph of maximum degree at most~4.
  It was observed in~\cite{twin-width6} that planar trigraphs with maximum (total) degree at most~9 have twin-width at most~9.
  This is because any planar graph has a pair of vertices on the same face with at most~9 neighbors (outside of themselves) combined~\cite{Kotzig55}.
  Hence we get a 9-sequence for $G$ that can be computed in polynomial time.
  Incidentally the twin-width of planar graphs (that is, planar trigraphs without red edge) but no restriction on the maximum degree is also at most~9~\cite{planar-tww}.
\end{proof}

Another very inapproximable is \ipath, which also does not admit a~polytime $n^{1-\varepsilon}$-approximation algorithm unless P$=$NP~\cite{Lund93}, and cannot be $r$-approximated in time $2^{o(n/r)}$ for any $r = o(n)$, unless the ETH fails~\cite{BonnetLP18}.
The non-induced version, the \lpath problem, has a notoriously big gap between the best known approximation algorithm whose factor is $n/\exp(\Omega(\sqrt{\log n}))$~\cite{Gabow08}, and the sharpest conditional lower bound which states that, for any $\varepsilon>0$, a~$2^{\log^{1-\varepsilon}n}$-approximation would imply that NP $\subseteq$ QP~\cite{Karger97}.

Despite being an open question for decades the existence or conditional impossibility of an approximation algorithm for \lpath with approximation factor, say, $\sqrt n$ has not been settled.
Nor do we know whether an $n^\varepsilon$-approximation for any $\varepsilon > 0$ is possible.
We now show that using our framework to obtain an $n^\varepsilon$-approximation for \ipath of \lpath in graphs of bounded twin-width is unlikely to work, in the sense that it would immediately yield such an approximation factor for \lpath in general graphs.

\begin{theorem}\label{thm:inapprox-path}
   For any $r=\omega(1)$, an $r$-approximation for \ipath or \lpath on graphs of twin-width at most~4 given with a 4-sequence would imply a~$(1+o(1))r$-approximation for \lpath in general graphs.
\end{theorem}
\begin{proof}
  It was shown in~\cite{Berge21} that any graph obtained by subdividing every edge of an $n$-vertex graph at least $2 \log n$ has twin-width at most~4.
  Besides, a 4-sequence can then be computed in polynomial time.

  Let $G$ be any graph with minimum degree at least~2 (note that this restriction does not make \lpath easier to approximate), and $G'$ be obtained from $G$ by subdividing each of its edges $2 \lceil \log n \rceil$ times, and let $s := 2 \lceil \log n \rceil + 1$.
  Let us observe that $G$ has a path of length $\ell$ if and only if $G'$ has a path of length $(\ell+2)s-2$ if and only if $G'$ has an induced path of length $(\ell+2)s-4$.
  Hence a polytime $r$-approximation for \ipath or \lpath in graphs of bounded twin-width given a 4-sequence would translate into a~$(1+o(1))r$-approximation for \lpath in general graphs.
\end{proof}

We can use~\cref{thm:inapprox-path} to get a similar weak obstruction to an $n^\varepsilon$-approximation for \mihp{\mathcal H} in graphs of bounded twin-width, for some infinite family of connected graphs $\mathcal H$.
Recall that by~\cref{lem:induced-matching-subexp} such an approximation algorithm does exist when $\mathcal H$ is a \emph{finite} collection of connected graphs.

Setting $\mathcal H$ to be the set of all paths does not serve that purpose, since one can then use the approximation algorithm for \mim.
Nevertheless this almost works.
We just need to \emph{decorate} the endpoints of the paths.
For every positive integer $n$, let $D_n$ be the \emph{decorated path} of length $n$, obtained from the $n$-vertex path $P_n$ by adding for each endpoint $u$ two adjacent vertices $u', u''$ both adjacent to $u$.
Informally, $D_n$ is a path terminated by a triangle at each end.

\begin{theorem}\label{thm:inapprox-mihp-dec-paths}
 Let $\mathcal H := \{D_n~:~n \in \mathbb N^+\}$ be the family of all decorated paths.
 If for every $\varepsilon >0$, \mihp{\mathcal H} admits an $n^\varepsilon$ on $n$-vertex graphs of bounded twin-width given with a 4-sequence, then so does \lpath on general graphs. 
\end{theorem}
\begin{proof}
  Let $G$ be any graph.
  For every pair $u \neq v \in V(G)$, define $G_{uv}$ as the graph obtained from $G$ by subdividing all its edges $2 \lceil \log (n+2) \rceil$ times, and adding two adjacent vertices $u', u''$ both adjacent to $u$, and two adjacent vertices $v',v''$ both adjacent to $v$.
  Since there are only two triangles in $G_{uv}$, only one graph of $\mathcal H$ can be present in a (mutually induced) packing.
  Thus \mihp{\mathcal H} is now equivalent to finding a longest path between $u$ and $v$.
  An $n^\varepsilon$-approximation algorithm for this problem would, by~\cref{thm:inapprox-path}, give a similar approximation algorithm for \lpath in general graphs.

  Despite $u',u'',v',v''$, $G_{uv}$ still admits a 4-sequence.
  For instance, first contract $u'$ and $u''$, and contract $v'$ and $v''$; this does not create red edges, and has the same effect as deleting $u''$ and $v''$.
  The obtained graph is an induced subgraph of a $2 \lceil \log (n+2) \rceil$-subdivision (of a graph on at most $n+2$ vertices).
  Hence it admits a polytime computable 4-sequence~\cite{Berge21}.
\end{proof}

\end{document}